\documentclass[a4paper,11pt]{article}
\usepackage{stmaryrd}
\usepackage{jheppub} 

\usepackage{amssymb}
\usepackage{amsmath}
\usepackage{amsthm}
\usepackage{float}
\usepackage{booktabs}
\usepackage{latexsym}
\usepackage{mathrsfs}
\usepackage{bm}
\usepackage{bm}
\usepackage[dvipsnames]{xcolor}
\usepackage{tikz}
\usepackage{newtxtext}
\usepackage{slashed}
\usepackage{bbm}
\usepackage{upgreek}
\usepackage{enumerate}

\DeclareMathOperator{\csch}{csch}
\theoremstyle{plain}
\newtheorem{proposition}{Proposition}
\newtheorem{lemma}{Lemma}
\newtheorem{theorem}{Theorem}

\newtheorem{corollary}{Corollary}

\newtheorem{remark}{Remark}

\def\bma{{\bm a}}

\def\bme{{\bm e}}

\def\bmg{{\bm g}}
\def\bmh{{\bm h}}
\def\bmi{{\bm i}}

\def\bml{{\bm l}}
\def\bmn{{\bm n}}
\def\bmm{{\bm m}}

\def\bmq{{\bm q}}

\def\bmv{{\bm v}}

\def\bmx{{\bm x}}

\def\bmzero{{\bm 0}}
\def\bmone{{\bm 1}}
\def\bmtwo{{\bm 2}}
\def\bmthree{{\bm 3}}

\def\bmA{{\bm A}}
\def\bmB{{\bm B}}
\def\bmC{{\bm C}}

\def\bmL{{\bm L}}

\def\bmN{{\bm N}}

\def\fraka{{\mathfrak a}}

\def\bmbeta{{\bm \beta}}

\def\bmdelta{{\bm \delta}}

\def\bmeta{{\bm \eta}}

\def\bmsigma{{\bm \sigma}}

\def\bmpartial{{\bm \partial}}
\def\bmnabla{{\bm \nabla}}

\usepackage{mathtools}
\usepackage{xparse}
\makeatletter
\newcommand{\raisemath}[1]{\mathpalette{\raisem@th{#1}}}
\newcommand{\raisem@th}[3]{\raisebox{#1}{$#2#3$}}
\makeatother
\NewDocumentCommand{\newrbar}{O{0pt} O{0pt}}{
  \ensuremath{\mathrlap{\raisemath{#2}{\hspace*{#1}{\mathchar'26\mkern-9mu}}}r}}

\newcounter{mnotecount}

\newcommand{\mnotex}[1]
{\protect{\stepcounter{mnotecount}}$^{\mbox{\footnotesize $\bullet$\themnotecount}}$ 
\marginpar{
\raggedright\tiny\em
$\!\!\!\!\!\!\,\bullet$\themnotecount: #1} }

\newcounter{mnote}

\usepackage{color}
\usepackage{xcolor}
\usepackage[normalem]{ulem}
\usepackage{soul}
\usepackage{hyperref}

\usepackage{cleveref}

\newcommand{\scri}[0]{\mathscr{I}}

\renewcommand{\leq}{\leqslant}
\renewcommand{\geq}{\geqslant}

\title{\boldmath Conformal geodesics and the hyperboloid at timelike infinity}

\author[a]{Achilleas V. Koukas,}
\author[b,c]{Mariem Magdy}
\author[d]{and Juan A. Valiente Kroon}

\affiliation[a]{Department of Physics, Princeton University,\\
Jadwin Hall, Washington Road, Princeton, New Jersey 08544, USA}
\affiliation[b]{Perimeter Institute for Theoretical Physics,\\
31 Caroline Street North, Waterloo, Ontario, N2L 2Y5, Canada}
\affiliation[c]{The Fields Institute for Research in Mathematical Sciences,\\
222 College St, Toronto, Ontario, M5T 3J1, Canada}
\affiliation[d]{School of Mathematical Sciences, Queen Mary, University of London,\\
Mile End Road, London E1 4NS, United Kingdom}

\emailAdd{ak2468@princeton.edu}
\emailAdd{mmagdy@perimeterinstitute.ca}
\emailAdd{j.a.valiente-kroon@qmul.ac.uk}

\abstract{A description of scattering in asymptotically flat spacetimes that includes both radiation and massive asymptotic states requires control of the relation between null and timelike infinity. Motivated by this problem, we use the properties of conformal geodesics to construct an extended representation of timelike infinity in Minkowski spacetime. The construction exploits the fact that every non-null metric geodesic in an Einstein spacetime is, up to an explicit local reparametrisation, a conformal geodesic. It replaces the point $i^+$ by a hyperboloidal boundary $\mathcal{H}^+$ meeting future null infinity $\mathscr{I}^+$ at a corner $\mathcal{C}^{+}$, while fixing the location of $\mathscr{I}^+$ in the associated conformal Gaussian gauge. In this gauge, we derive the equations governing massless scalar and spin-$2$ fields. In particular, we obtain a transport equation for the scalar field along future null infinity and analyse the asymptotic behaviour of its spherical-harmonic modes near timelike infinity. We also derive the explicit massless spin-$2$ equations in this gauge. Finally, we examine the relation to projective extensions and to the hyperbolic base of the extended Carrollian space $\mathrm{Ti}$. Our results provide a geometric framework for studying the connection between late-time and null asymptotics which is a prerequisite for scattering constructions relevant to flat space holography.}

\begin{document}
\maketitle
\flushbottom

\section{Introduction}
Flat space holography \cite{Racl21,DonFioHerRuc22,Donnay23} seeks a boundary description of quantum gravity in asymptotically flat spacetimes in which bulk scattering amplitudes are encoded by boundary observables. In this programme, the asymptotic geometry determines where scattering data are represented and how the symmetries acting on incoming and outgoing states are related. Incorporating both massless radiation and massive asymptotic states requires a description of null and timelike infinity as well as a relation between them. Such a geometric framework is therefore relevant both to the formulation of scattering theory and to the boundary structures underlying a holographic dictionary.

In the standard Penrose conformal description of Minkowski spacetime, future timelike infinity is represented by a single point $i^{+}$ which does not directly resolve the dependence of massive-field asymptotic data on timelike directions. Replacing $i^{+}$ by an extended space makes this dependence accessible to geometric and analytic study. Similar considerations motivated earlier constructions resolving spatial infinity $i^{0}$ \cite{AshHan78,AshRom92,Fri81a}. Inspired by Friedrich's approach to the blow-up of spatial infinity, our main goal is to construct an extended representation of timelike infinity that meets future null infinity at a corner.

Despite historically receiving less attention than null infinity or spatial infinity, there is, nevertheless, a substantial amount of work on the subject aimed at developing an understanding of the behaviour of fields (including gravitation) at $i^+$ and, at the same time, providing alternative descriptions that allow to disentangle the structures associated to this region of spacetime. Much of this work falls into three categories:

\begin{enumerate}[(a)]
    \item \emph{the analysis of the conformal and asymptotic structure of} $i^+$ with the aim of clarifying whether timelike infinity can be given enough geometric structure to describe the asymptotic state of an isolated gravitating system or the scattering of massive fields ---see e.g.\cite{Por82,Per82a,Cut89,GenShi98}. Early work along this line tried to extract or attach additional structure to timelike infinity ---e.g. direction dependent limits, projective structures, hyperboloidal completions so that the late-time state of the systems under consideration could be described geometrically rather than along curves in the spacetime.
    
    \item \emph{The analysis of asymptotic symmetries and conserved quantities at $i^+$} aimed at establishing the connection of these with their counterparts at null infinity $\mathscr{I}^+$ and understanding gravitational memory ---see e.g. \cite{SolBha25,ComGraWei23,ChaGhoHoqKhaVir22,GenShi98}.
    
    \item \emph{The analytic and microlocal treatments aimed at turning $i^+$} into a geometric boundary at which the long-time behaviour of solutions to hyperbolic equations, in particular the Einstein field equations, can be described and used in scattering constructions ---see e.g. \cite{HinVas20,BorHer25,BorChaHer25,FigHavProSal22}. 
\end{enumerate}

Of particular relevance to our goal are the extended descriptions of timelike infinity developed in gravitational scattering and flat space holography ---see e.g. \cite{FigHavProSal22,HaveNgyProSal24,ComGraWei23}. In Minkowski spacetime, the asymptotic four-velocities of massive particles are naturally parametrised by the unit hyperboloid $\mathbb{H}^3$, which can also be viewed as Euclidean anti-de Sitter space \cite{CampigliaLaddha15}. The extended timelike space $\mathrm{Ti}$ is a line bundle over this hyperboloid endowed with a Carrollian structure \cite{FigHavProSal22,HaveNgyProSal24,BorHer25,BorChaHer25}. Earlier work \cite{BorHer25} extended this description to curved spacetimes under an implicit assumption of vanishing mass aspect. This restriction is removed in \cite{BorChaHer25} where an invariant construction is given for spacetimes satisfying a notion of asymptotic flatness at timelike infinity formulated through regularity conditions on a suitably rescaled inverse metric and normal to the boundary. The line bundle $\mathrm{Ti}$ is constructed over the resulting timelike boundary. The fibres of $\mathrm{Ti}$ parametrise the freedom in the second-order term of the boundary defining function.

These constructions draw on the seminal work of Ashtekar and Hansen \cite{AshHan78}, in which spatial infinity is described through direction-dependent limits of fields at $i^{0}$. The unit spacelike directions of approach form a hyperboloid with three-dimensional de Sitter geometry. The associated extended space $\mathrm{Spi}$ is a line bundle over this hyperboloid and provides a geometric model for the later constructions at timelike infinity. Ashtekar and Romano \cite{AshRom92} subsequently formulated spatial infinity directly as a timelike boundary hypersurface of a completed spacetime and encoded the asymptotic gravitational field in regular tensor fields on this boundary.

An important limitation of the original Ashtekar-Hansen and Ashtekar-Romano constructions \cite{AshHan78,AshRom92} is that they do not provide a joint description of spatial and null infinity that controls the regularity of fields where these asymptotic regions meet. A similar limitation arises in the later constructions of extended timelike boundaries \cite{FigHavProSal22,HaveNgyProSal24,BorHer25,BorChaHer25} which do not by themselves establish the regularity of fields at their intersection with null infinity. In frameworks for gravitational scattering, matching between different asymptotic regions relies on additional assumptions about the regularity and asymptotic behaviour of the fields --- see e.g. \cite{ComGraWei23}. A joint description of these asymptotic regions makes it possible to investigate whether these conditions are satisfied by solutions of the Einstein field equations. This work then provides a setting for addressing this question, particularly for the intersection of future null and timelike infinity, with potential applications to the description of massive asymptotic states in flat space holography.

\medskip
\noindent
\textbf{Spacetimes endowed with $i^+$.} A foundational result for addressing this problem is Friedrich's semiglobal stability for the Minkowski spacetime \cite{Fri86b} which ensures that small enough perturbations of hyperboloidal suitably smooth initial data for the Minkowski spacetime develop into a future geodesically complete semiglobal spacetime with an asymptotic structure similar to that of the Minkowski spacetime. In particular, the spacetime can be smoothly conformally extended towards the future of the initial hyperboloid. Moreover, the generators of future null infinity, $\mathscr{I}^+$, are complete and intersect at a point $i^+$ ---future timelike infinity. The point $i^+$ is of particular physical significance as it represents the future endpoint of timelike geodesics.

Friedrich's stability theorem is based on the hyperboloidal initial value problem for the (conformal) Einstein field equations, first considered in \cite{Fri83}. In this article, we take the spacetimes arising from this initial value problem as the setting for our analysis. Their future timelike geodesic completeness and smooth Penrose conformal completion, including $i^+$, provide an ideal setting for constructing an extended representation of timelike infinity and investigating its relation to $\mathscr{I}^+$.

\medskip
\noindent
\textbf{Conformal geodesics.} The analysis in this article makes use of conformal geodesics to resolve the structure of timelike infinity thus providing a blow-up in the form of what we call the \emph{hyperboloid at timelike infinity},  $\mathcal{H}^+$. Conformal geodesics are a generalisation of the standard notion of metric geodesics with suitable transformation properties upon a conformal rescaling ---and more generally affine changes of the connection. Conformal geodesics were first considered in the context of General Relativity by Schmidt \cite{Sch86} ---see also \cite{FriSch87}. One key property of these geometric objects is that they provide a canonical way of constructing conformal extensions of spacetime.

Conformal geodesics naturally complement Friedrich's extended conformal Einstein field equations \cite{Fri95,Fri98a,CFEBook} by providing a geometric way of constructing gauges to analyse the evolution of the gravitational field by means of conformal methods ---the so-called \emph{conformal Gaussian gauge systems}. Crucially, for vacuum spacetimes with or without Cosmological constant, conformal Gaussian gauge systems allows to decouple the knowledge of the particular representative of the conformal class of the conformal metric from the evolution of the geometric fields describing the spacetime. In particular, the associated conformal factor is known explicitly once the gauge has been set up ---thus giving \emph{a priori} knowledge of the locus of the conformal boundary. This idea has been central in the formulation of an initial boundary value problem to construct anti-de Sitter-like spacetimes \cite{Fri95}. It can also be used to provide a simpler, more transparent, proof of the classic results \cite{Fri86b} on the global non-linear stability of the de Sitter spacetime and the semiglobal stability of the Minkwoski spacetime from hyperboloidal data ---see \cite{LueVal09,CFEBook}. 

Conformal geodesics in Einstein spacetimes, i.e. those for which the Ricci tensor is proportional to the metric, have special properties. In particular, every metric geodesic in an Einstein spacetime can be reparametrised as a conformal geodesic. In this article, we exploit this observation to reparametrise timelike geodesics in the Minkowski spacetime and suitable nonlinear perturbations thereof to obtain a congruence of conformal geodesics terminating at timelike infinity. This congruence allows, in turn, to set up a conformal Gaussian system which covers the domain of dependence $D^+(\mathcal{H}_\star)$ of a hyperboloid $\mathcal{H}_\star$. The reparametrisation of the metric geodesics attaches a future endpoint to the members of the family of curves under consideration. This construction, however, does not lead to a \emph{projective compactification} as it is linked to a specific type of timelike geodesics. However, as a consequence of the use of Gaussian coordinates, in this representation, the point $i^+$ has been blown up to a hyperboloid ---the \emph{hyperboloid at timelike infinity} $\mathcal{H}^+$. Thus, this construction provides a \emph{conformal compactification} of $D^+(\mathcal{H}_\star)$ for which the associated metric is degenerate but smooth.

\medskip
In order to gain intuition, all the constructions described in the previous paragraph are first analysed in the context of the (exact) Minkowski spacetime. Friedrich's semiglobal stability result allows, in turn, to readily extend these constructions to nonlinear perturbations ---i.e. dynamic spacetimes with, in general, non-vanishing fluxes of energy across null infinity.

\medskip
\noindent
\textbf{Evolution equations near $i^+$.} A principal motivation of the present analysis is to develop a framework for the analysis of evolution equations (including those governing the gravitational field) near timelike infinity. In particular, our construction gives rise to a manifold with a corner $\mathcal{C}^+$ ---this corresponds to the locus of points at which null infinity $\mathscr{I}^+$ and the hyperboloid at timelike infinity $\mathcal{H}^+$ intersect. Manifolds with corners at timelike infinity already occur in the microlocal stability analysis of Hintz and Vasy \cite{HinVas20} and in the scattering framework of \cite{ComGraWei23}. Our construction brings this perspective into a gauge adapted to Friedrich's extended conformal Einstein field equations.

We provide an initial analysis using the scalar wave equation on the rescaled spacetime. In particular, we derive its transport equation along $\mathscr{I}^+$ and analyse the asymptotic behaviour of the spherical-harmonic modes of its homogeneous part towards $\mathcal{C}^+$. Since the full transport equation also involves transverse derivatives, these mode asymptotics alone do not determine the full boundary evolution. We also briefly consider the wave equations describing the massless spin-$2$ field describing linearised gravity in this representation. A detailed analysis of this system and the treatment of massive fields lie beyond the scope of the present article.

\medskip
\noindent
\textbf{$\mathscr{I}$-fixing.} One of the consequences of the use of a conformal Gaussian system based on timelike curves ending at $i^+$ together with Friedrich's \emph{extended conformal Einstein field equations} is that one obtains a setting for the analysis of non-linear perturbations of the Minkowski spacetime for which the location of null infinity $\mathscr{I}^+$ is known \emph{a priori}  as the locus of points for which a \emph{spatial coordinate} vanishes. This type of setting has been much sought-after in hyperboloidal approaches to numerical relativity and is often known as \emph{$\mathscr{I}$-fixing} ---see e.g. \cite{Zen08b}. An alternative approach to $\mathscr{I}$-fixing for a conformal representation of the Einstein field equations has been recently considered in the construction of perturbations of the Minkowski spacetime with smooth conformal extensions by N\"utzi \cite{Nut25}. 

\medskip
The results presented here provide a first step towards a broader analysis of timelike infinity within the framework of the extended conformal Einstein field equations. Further potential developments include applications to the analysis of the hyperboloidal initial value problem for the Einstein field equations using techniques from microlocal analysis, numerical implementations of the constructions developed here, and their generalisation to the domain of outer communication of black hole spacetimes. These directions will be considered elsewhere.

\subsection*{Structure of the article}

In Section \ref{Section:Tools}, we provide a brief summary of the ideas and tools of conformal and projective geometry that will be used in the rest of the article. This is intended mostly as a quick reference. The reader interested in full details and proof is referred to the indicated bibliography. Section \ref{Section:TimelikeGeodesicsMinkowski} provides a discussion of conformal and projective properties of timelike geodesics in Minkowski. Particular focus is given to the way in which these curves can be used to construct representations of timelike infinity. This section also provides a first example of a \emph{blow-up} of timelike infinity by means of timelike geodesics. This example will serve as model for the rest of the article. We also discuss the connection of this construction to projective extensions. Section \ref{Section:Conformal Gaussian Systems} is the main one of this article and discusses the blow-up of timelike infinity in the Minkowski spacetime by means of conformal Gaussian systems based on a fiduciary hyperboloid. For reference, a brief discussion of the construction of hyperboloidal hypersurfaces by means of height functions is provided. Moreover, we identify two different specialisations of conformal Gaussian systems: namely one in which the hyperboloids are described by means of the so-called Poincar\'e ball model of hyperbolic space and the other an approach introduced by A. Zenginoglu. In Section \ref{Section:MasslessFields}, we apply the constructions of the previous section to study the behaviour of the wave equation near timelike infinity and along null infinity. We also make some general remarks on the structure of the massless spin-2 equations (linearised gravity) in this context. Section \ref{Section: projective-extnesion-Ti} explores the connection of the representations of timelike infinity constructed in Section \ref{Section:Conformal Gaussian Systems} with projective extensions and the so-called $\mbox{Ti}^+$ asymptote. In Section \ref{Section:NonlinearPerturbations}, we prove a stability result for our construction which is summarised in Theorem \ref{Proposition:StabilityGauge}. More precisely, it is shown that conformal Gaussian systems analogous to the ones in Section \ref{Section:Conformal Gaussian Systems} can be constructed to non-linear perturbations of the Minkowski spacetime. This result is essential if one is to turn the constructions discussed in this article into a tool to study the asymptotics of the gravitational field. Section \ref{Section:Conformal Gaussian Systems} provides some conclusions and an outlook for future work. The article contains, additionally three appendices in which further details of some of the main calculations in the article are provided.

\subsection*{Conventions and notation}

Our conventions are consistent with \cite{CFEBook}. In particular, our
metric signature is $(+, -, -, -)$, and the Riemann curvature tensor
associated with the Levi-Civita connection of a metric $g_{ab}$ is
defined by \mbox{$[\nabla_a, \nabla_b] u^d = R^d_{\phantom{d}cab}
u^c$}. For a given spin dyad $\{ o^A, \iota^A \}$ with $o_A \iota^A =
1$, we write $\epsilon_\bmA{}^A$, $\bmA \in \{ \bmzero, \bmone \}$, to denote
$\epsilon_{\bmzero}{}^A = o^A$ and $\epsilon_{\bmone}{}^A = \iota^A$. Spinorial
indices are raised and lowered using the antisymmetric
$\epsilon$-spinor $\epsilon_{AB} = o_A \iota_B - \iota_A o_B$ (with
inverse $\epsilon^{AB} = o^A \iota^B - \iota^A o^B$), e.g. $\xi_B =
\xi^A \epsilon_{AB}$, using the convention that contracted indices
should be \emph{adjacent, descending to the right}. As usual, the
spacetime metric $g_{ab}$ decomposes as $g_{ab} = \epsilon_{AB}
\bar{\epsilon}_{A' B'}$, where $\bar{\epsilon}_{A'B'} =
\overline{\epsilon_{AB}}$. The spin dyad $\epsilon_\bmA{}^A$ gives
rise to a tetrad of null vectors $\bme_{\bmA \bmA'} = \bme_{\bmA
\bmA'}{}^{AA'} \partial_{AA'} = \epsilon_\bmA{}^A
\bar{\epsilon}_{\bmA'}{}^{A'} \partial_{A A'}$. The spin connection
coefficients $\Gamma_{\bmA \bmA'}{}^\bmB{}_\bmC$ are then defined as
$\Gamma_{\bmA \bmA'}{}^\bmB{}_\bmC = - \epsilon_\bmC{}^Q \nabla_{\bmA
\bmA'} \epsilon^\bmB{}_Q$, where $\nabla_{\bmA \bmA'} = \bme_{\bmA
\bmA'}{}^{AA'} \nabla_{A A'}$. In Newman--Penrose language, the spin
connection coefficients $\Gamma_{\bmA \bmA' \bmB \bmC} = \Gamma_{\bmA
\bmA' \bmC \bmB}$ are exactly the spin coefficients $\gamma_{\bmA
\bmA' \bmB \bmC}$ (see \cite{PenRin86}, Summary of Vol. 1). When
integrating over $\mathrm{SU}(2)$, $\mu$ denotes the normalized Haar
measure on $\mathrm{SU}(2)$. Throughout the majority of the paper, we
will be working on \emph{rescaled} Minkowski space (in the F-gauge),
and we shall denote objects (fields, connections, spin dyad) on this
spacetime plainly. When working on physical Minkowski space, we will
denote objects with a tilde, e.g. $\tilde{\phi}$.

\section{Tools of conformal and projective geometry}
\label{Section:Tools}

In this section, we provide a brief summary of the notions of conformal and projective geometry that will be used in this article. For further details and proofs, see e.g. \cite{CFEBook}.

\subsection{Metric, Weyl and projective connections}
In this article, we will make use of various types of connections on Lorentzian manifolds. We briefly review the main definitions and properties.

\subsubsection{Changes of connection}
In the following let $\tilde\bmnabla$, $\bmnabla$ denote, respectively, the Levi-Civita connection of two (Lorentzian) metrics $\tilde\bmg$ and $\bmg$ over a (4-dimensional) manifold $\tilde{\mathcal{M}}$. By definition, one has that
\begin{eqnarray*}
&\tilde\bmnabla \tilde\bmg =0, \qquad \mathrm{Torsion}[\tilde\bmg]=0, & \\
& \bmnabla\bmg =0, \qquad \mathrm{Torsion}[\bmg]=0.
\end{eqnarray*}
We will say that the two metrics $\tilde\bmg$ and $\bmg$ are \emph{conformally related} if there exists a conformal factor $\Xi$ such that
\begin{equation}
\bmg =\Xi^2 \tilde{\bmg}.
\label{ConformalRescaling}
\end{equation}
The associated transformation for the connections is then given by
\[
\bmnabla -\tilde\bmnabla =\mathbf{S}\big(\Xi^{-1}\mathbf{d}\Xi  \big),
\]
where $\mathbf{S}$ denotes the rank 4 tensor which in abstract index notation is given by
\[
S_{ab}{}^{cd} = \delta_a{}^c \delta_b{}^c + \delta_a{}^d \delta_b{}^c - \tilde{g}_{ab}\tilde{g}^{cd}.
\]
The collection of metrics conformally related to $\tilde\bmg$ constitutes the \emph{conformal class} of $\tilde\bmg$ and is denoted by $[\tilde\bmg]$. 

A Weyl connection $\hat\bmnabla$ is a torsion-free connection which is not necessarily metric but which, nevertheless, respects the conformal class. More precisely, one requires that 
\[
\hat{\nabla}_a\tilde{g}_{ab} = -2 f_a \tilde{g}_{bc} 
\]
where $f_a$ is a 1-form. 

\medskip
Now, two torsion-free affine connections $\bmnabla$ and $\grave{\bmnabla}$ are said to be \emph{projectively equivalent} if there exists a 1-form $\bm \varpi$ such that 
\[
\grave{\nabla}_a v^b - \nabla_a v^b = \delta_a{}^b \varpi_c v^c+ v^b \varpi_a.
\]
The latter, in turn, implies the transformation rule
\begin{equation}
\grave{\Gamma}_\mu{}^\nu{}_\lambda -\Gamma_\mu{}^\nu{}_\lambda = \delta_\mu{}^\nu \varpi_\lambda + \delta_\lambda{}^\nu \varpi_\mu, 
\label{ProjectiveConnection}
\end{equation}
for given local coordinates $x=(x^\mu)$. A projective change of connection preserves geodesics in the sense that if the tangent vector $\bmx'$ to a curve $x(s)$, $s\in I\subseteq\mathbb{R}$, satisfies $\nabla_{\bmx'}{\bmx}'=0$, then 
\[
\grave{\nabla}_{{\bmx}'}{\bmx}' = 2 \langle {\bm \varpi},{\bmx}'\rangle \,{\bmx}',
\]
where $\langle{\bm \varpi},{\bmx}'\rangle$ denotes the action of the 1-form $\bm \varpi$ on the vector $\bmx'$ ---i.e. $\varpi_a x^{\prime a}$ in abstract index notation.

\subsubsection{Extensions and compactifications}
Recall that, given a spacetime $(\tilde{\mathcal{M}},\tilde{\bmg})$, a \emph{conformal extension} consists of a spacetime $(\mathcal{M},\bmg)$ together with an embedding of $\tilde{\mathcal{M}}$ into $\mathcal{M}$ such that the relation \eqref{ConformalRescaling} holds. In addition, the conformal factor $\Xi$ is required to be a \emph{boundary defining function} for $\mathcal{M}$, with $\Xi=0$ on $\partial\mathcal{M}$, while the conformal metric $\bmg$ extends regularly to $\partial\mathcal{M}$.

\begin{remark}
{\em Following standard usage, in the following, it will be assumed that the metric $\tilde{\bmg}$ satisfies the vacuum Einstein field equations and call $(\tilde{\mathcal{M}},\tilde\bmg)$ the \emph{physical spacetime}. The pair $(\mathcal{M},\bmg)$ will be referred to as the \emph{unphysical spacetime}.}
\end{remark}

Now, given a boundary defining function $\varsigma$, one speaks of a \emph{projective extension} of $(\tilde{\mathcal{M}},\tilde{\bmg})$ with boundary at the locus of points for which $\varsigma=0$ if the connection extends smoothly through the boundary ---that is, if after setting
\[
{\bm \varpi} = \frac{\mathbf{d}\varsigma}{\varsigma}
\]
in formula \eqref{ProjectiveConnection},
the connection coefficients $\grave{\Gamma}_\mu{}^\nu{}_\lambda$ in local coordinates $x=(x^\mu)$ given by
\begin{equation}
\grave{\Gamma}_\mu{}^\nu{}_\lambda = \Gamma_\mu{}^\nu{}_\lambda + \frac{1}{\varsigma}\left( \delta_\mu{}^\nu \partial_\lambda \varsigma + \delta_\lambda{}^\nu \partial_\mu\varsigma \right)
\label{ProjectiveExtension:ChristiffleTransformation}
\end{equation}
are smooth up to the boundary.

\subsection{Conformal geodesics in Einstein spacetimes}
Recall that given a curve $x(s)$, with $s\in I\subset \mathbb{R}$ with
tangent $\bmx'(s)$ is said to be a (metric) geodesic if it satisfies the
equation
\[
\tilde{\nabla}_{\bmx'}\bmx' =0. 
\]

\medskip
A \emph{conformal geodesic} on a spacetime $(\tilde{\mathcal{M}},\tilde{\bmg})$
is a pair $(x(\tau),\bmbeta(\tau))$ consisting of a curve $x(\tau)$ on
$\tilde{\mathcal{M}}$, $\tau\in I\subset \mathbb{R}$, with tangent
$\dot{\bmx}(\tau)$ and a covector $\bmbeta(\tau)$ along $x(\tau)$
satisfying the equations
\begin{subequations}
\begin{eqnarray}
&& \tilde{\nabla}_{\dot{\bmx}} \dot{\bmx} = -2 \langle
   \bmbeta,\dot{\bmx}\rangle \dot{\bmx} +
   \tilde{\bmg}(\dot{\bmx},\dot{\bmx}) \bmbeta^\sharp, \label{CGEqn1}\\
  && \tilde{\nabla}_{\dot{\bmx}}\bmbeta =\langle \bmbeta,
     \dot{\bmx}\rangle\bmbeta -\frac{1}{2}\tilde{\bmg}^\sharp
     (\bmbeta,\bmbeta)\dot{\bmx}^\flat +
     \tilde{\bmL}(\dot{\bmx},\cdot), \label{CGEqn2}
\end{eqnarray}
\label{CGEqn}
\end{subequations}
where $\tilde{\bmL}$ denotes the \emph{Schouten tensor} of the
Levi-Civita connection $\tilde{\bmnabla}$. For a solution to the vacuum Einstein field equations one 
has 
\[
\tilde{\bmL} =0 .
\]

In the present article, we will rely heavily on the following result
---see Lemma 5.2 in \cite{CFEBook}:

\begin{lemma}[Standard geodesics as conformal geodesics]
\label{Lemma:GeodesicsAsConformalGeodesics}
Any non-null $\tilde{\bmg}$-geodesic in an Einstein spacetime
$(\tilde{\mathcal{M}},\tilde{\bmg})$ is, up to a reparametrisation, a
non-null conformal geodesic.
\end{lemma}

\begin{proof}
Let $x(s)$ denote a solution to the metric geodesic equation
$\tilde{\nabla}_{\bmx'} \bmx' =0$. Consider now a reparametrisation of
the curve of the form $s=s(\tau)$. We complete $x(\tau)$ to a
conformal geodesic using an Ansatz of the form
\[
\bmbeta =\alpha(\tau)
  \bmx^{\prime \flat}.
\]
 Writing $\dot{\bmx} = \dot{s} \bmx^\prime$, it
  follows readily that equation \eqref{CGEqn1} implies
\begin{equation}
  \ddot{s} + \alpha \dot{s}^2 \tilde{\bmg}(\bmx',\bmx')=0,
  \label{EquationParameterS}
\end{equation}
where it is noticed that $ \tilde{\bmg}(\bmx',\bmx')$ is constant
along the curve as it is $\tilde{\bmg}$-geodesic. Now, substituting
the Ansatz for $\bmbeta$ into equation \eqref{CGEqn2}, one concludes
that
\begin{equation}
\dot{\alpha} = \frac{1}{2}\alpha^2 \dot{s} \tilde{\bmg}(\bmx',\bmx') .
\label{EquationFactorAlpha}
\end{equation}
Equations \eqref{EquationParameterS} and \eqref{EquationFactorAlpha}
constitute a system of ordinary differential equations for $s(\tau)$
and $\alpha(\tau)$ which can always be solved ---at least locally. 
\end{proof}

We shall also make use of the following result generalising  Proposition 5.1 in \cite{CFEBook}:

\begin{proposition}[the conformal factors associated to a
  conformal geodesic]
  \label{Lemma:ConformalGaussianGauge}
Let $(\tilde{\mathcal{M}},\tilde{\bmg})$ denote an Einstein spacetime. Suppose
that $(x(\tau),\bmbeta(\tau))$ is a solution to the conformal geodesic
equations \eqref{CGEqn1}-\eqref{CGEqn2} and that $\{\bme_\bma\}$ is a
$\bmg$-orthonormal frame propagated using Weyl propagation. If $\bmg
=\Theta^2 \tilde{\bmg}$ is such that 
\[
\bmg(\dot{\bmx},\dot{\bmx})=\varkappa^2,
\]
with $\varkappa$ a scalar satisfying 
\[
\tilde{\nabla}_{\dot{\bmx}}\varkappa=0
\]
then the conformal factor $\Theta$ satisfies
\[
\Theta (\tau) = \Theta_\star + \dot{\Theta}_\star(\tau-\tau_\star) +
\frac{1}{2}\ddot{\Theta}_\star (\tau-\tau_\star)^2,
\]
where the coefficients $\Theta_\star \equiv \Theta(\tau_\star)$,
$\dot{\Theta}_\star \equiv \dot{\Theta}(\tau_\star)$ and
$\ddot{\Theta}_\star \equiv \ddot{\Theta}(\tau_\star)$ are constants
along the conformal geodesic and are subject to the constraints
\[
\dot{\Theta}_\star = \langle
\bmbeta_\star,\dot{\bmx}_\star\rangle\Theta_\star, \qquad \Theta_\star
\ddot{\Theta}_\star = \frac{\varkappa^2}{2}\tilde{\bmg}^\sharp
(\bmbeta_\star,\bmbeta_\star).
\]
Further, along each conformal geodesic
\[
\varkappa \Theta \beta_\bmzero = \dot{\Theta}, \qquad \Theta\beta_\bmi =
\Theta_\star (\beta_\bmi)_\star,
\]
where $\beta_\bma \equiv \langle \bmbeta,\bme_\bma \rangle$.
\label{Proposition:Canonical-conformal-factor}
\end{proposition}

\begin{proof}
    The proof is similar to that of Proposition 5.1 in \cite{CFEBook} in which the particular case $\varkappa=1$ was discussed.
\end{proof}

\begin{remark}
{\em A frame $\{ \bme_\bma \}$ is said to be \emph{Weyl propagated} if it satisfies the propagation equation
\[
\tilde{\nabla}_{\dot{\bmx}} \bme_{\bma} = -\langle \bmbeta,\bme_\bma \rangle \dot{\bmx} - \langle \bmbeta,\dot{\bmx} \rangle \bme_\bma + \tilde{\bmg}(\bme_\bma,\dot{\bmx} )\bmbeta^\sharp.
\]
In particular, setting $\bme_\bmzero=\dot{\bmx}/\varkappa$ recovers the conformal geodesic equation \eqref{CGEqn1}. Weyl propagation preserves the orthogonality of a  set of vectors. 
}
\end{remark}

\section{Timelike geodesics in the Minkowski spacetime as conformal geodesics}
\label{Section:TimelikeGeodesicsMinkowski}

In this section, we make use of Lemma \ref{Lemma:GeodesicsAsConformalGeodesics} to recast a specific class of metric geodesics in the Minkowski spacetime as conformal geodesics. This construction, in turn, is used to obtain a conformal representation of the neighbourhood of timelike infinity. 

\subsection{Basic definitions and equations}
In the following, let $(\mathbb{R}^{1,3},\tilde{\bmeta})$ denote the Minkowski spacetime. We consider the Minkowski line element in spherical
coordinates ---namely,
\begin{equation}
    \tilde{\bmeta} = \mathbf{d}t\otimes \mathbf{d}t - \mathbf{d}r \otimes
\mathbf{d} r - r^2 \slashed \bmsigma,
\label{Minkowski-spherical}
\end{equation}
with $\slashed \bmsigma$ the round metric on $\mathbb{S}^2$. The elements of the family of
curves given by
\begin{equation}
\gamma_{(r_\star,\theta_\star,\varphi_\star)}\equiv (s,
r_\star,\theta_\star,\varphi_\star), \qquad s\in \mathbb{R}
\label{TimelikeGeodesics}
\end{equation}
are standard (i.e. metric)  geodesics with tangent
$\bmx'(s)=\bmpartial_t$ ---in the following, the dash
$\phantom{i}^\prime$ will denote differentiation with respect to the parameter $s$. When plotted in a Penrose diagram, these curves have
endpoints at the points $i^\pm$ ---future and past timelike infinity; see  Figure \ref{Timelike-geodesics-to-infinity}.

\begin{remark}
\label{Remark:BreakinLorentzInvariance}
{\em The choice of this particular congruence singles out the timelike direction $\bmpartial_t$ and is therefore not invariant under the full Lorentz group. Given
$$
\mathcal{D}\equiv (\bmpartial_t)^\perp
=\left\{\bmv \in T\mathbb{R}^{1,3}\; |\;  \tilde{\bmeta}(\bmpartial_t,\bmv)=0 \right\},
$$
a generic Lorentz transformation $\Lambda$ preserves orthogonality, but maps this distribution according to
$$
\Lambda(\mathcal{D})=(\Lambda\bmpartial_t)^\perp.
$$
For spatial rotations, $\Lambda\bmpartial_t=\bmpartial_t$, so that $\mathcal{D}$, and hence the congruence, is preserved. By contrast, a generic boost satisfies $\Lambda\bmpartial_t\neq\bmpartial_t$ and therefore maps the congruence to one associated with a different timelike direction. Generally, a generic Lorentz transformation would map the congruence $\gamma_{(r_\star,\theta_\star,\varphi_\star)}$ to a more general family of timelike geodesics with non-vanishing velocity as $s \to \infty$. Further discussion of the consequences of this observation is given in subsection \ref{Subsubsection:CauchyProjectiveExtension} and Section \ref{Section: projective-extnesion-Ti}. }
\end{remark}

\subsection{The reparametrisation}
\label{Subsection:Reparametrisation}
We follow the procedure outlined in Lemma 5.2 in \cite{CFEBook} to
cast the curves $\gamma_{(r_\star,\theta_\star,\varphi_\theta)}$ as
conformal geodesics. To this end, we consider a reparametrisation of
the form
\[
s\mapsto s(\tau),
\]
so that for a curve $x(s)$, one has that
\[
x(s) \mapsto x(\tau) = x(s(\tau)).
\]
In the following, we denote by $\dot{\phantom{X}}$ the derivative with
respect to the parameter $\tau$. In particular, for the tangent vector of the
curves $x(\tau)$, one writes $\dot{\bmx}(\tau)$. Then, one has that
$\dot{\bmx}=\dot{s}\bmx'$. The curve $x(\tau)$ is
supplemented by a 1-form $\bmbeta(\tau)$ using the Ansatz
\[
\bmbeta = \alpha(\tau) \bmx^{\prime\flat}.
\]
The functions $s(\tau)$ and $\alpha(\tau)$ satisfy the equations
\begin{subequations}
\begin{eqnarray}
  && \ddot{s} +\alpha \dot{s}^{ 2} =0, \label{ReparametrisationEquation1}\\
  && \dot{\alpha} = \frac{1}{2}\alpha^2 \dot{s}, \label{ReparametrisationEquation2}
\end{eqnarray}
\end{subequations}
where it has been used that
$\tilde{\bmeta}(\bmx^\prime,\bmx^\prime)=1$. The general solution to the above system of ordinary differential
equations is given by
\begin{subequations}
\begin{eqnarray}
&& s(\tau) -s_\star =
   \frac{2(\tau -\tau_\star)}{2\tau'_\star+\alpha_\star (\tau-\tau_\star)}, \label{GeneralReparametrisation1}\\
&& \alpha(\tau)-\alpha_\star  = \frac{\alpha^2_\star}{2\tau'_\star}(\tau-\tau_\star), \label{GeneralReparametrisation2}
\end{eqnarray}
\end{subequations}
where $s_\star$ is a fiduciary value of the affine parameter of the
geodesics and, moreover, we have set
\[
\tau_\star \equiv \tau(s_\star), \qquad \tau'_\star \equiv
\tau'(s_\star), \qquad \alpha_\star \equiv \alpha(\tau_\star).
\]
The derivation of this solution is provided in Appendix \ref{Appendix:Solving the reparametrisation equations}. For future reference, it is observed that the relation
\eqref{GeneralReparametrisation1} can be inverted to yield
\[
\tau-\tau_\star = \frac{2\tau'_\star(s-s_\star)}{2+ \alpha_\star
  (s_\star- s)}.
\]
From this it follows, using the L'Hopital rule, that
\[
\lim_{s\rightarrow \infty} (\tau-\tau_\star)=-\frac{2\tau'_\star}{\alpha_\star}. 
\]

\begin{remark}
    {\em A close inspection of the derivation of equations \eqref{GeneralReparametrisation1}-\eqref{GeneralReparametrisation2} shows that \emph{they are independent of the metric}. Thus, these relations are universal for vacuum spacetimes. This observation will be explored further elsewhere and opens the tantalising possibility of applying the methods discussed in this article to a wider class of spacetimes.}
\end{remark}

\subsubsection{Relation to the parameters of the canonical conformal
  factor}
It will be useful to relate the integration constants $\tau'_\star$
and $\alpha_\star$ to the parameters of the canonical conformal factor
$\Theta$ introduced in Proposition \ref{Lemma:ConformalGaussianGauge}. To this end, recall that for an Einstein spacetime, one has 
\begin{subequations}
    \begin{eqnarray}
        && \dot{\Theta}_\star =\langle \bmbeta_\star, \dot{\bmx}_\star \rangle
            \Theta_\star, \\
        && \Theta_\star \ddot{\Theta}_\star =\frac{\varkappa^2}{2}\tilde{\bmg}^\sharp(\bmbeta_\star,\bmbeta_\star)  
    \end{eqnarray}
    \label{Conformal-factor-data-vacuum}
\end{subequations}
---see Lemma \ref{Lemma:ConformalGaussianGauge}. It is also recalled that
\[
\bmg = \Theta^2 \tilde{\bmg}, \qquad \bmg^\sharp =\Theta^{-2} \tilde{\bmg}^\sharp,
\]
and that the parameters $s$ and $\tau$ are, respectively,
$\tilde{\bmg}$ and $\bmg$ proper times so that
\[
\tilde{\bmg}(\bmx',\bmx') =1, \qquad \bmg(\dot{\bmx},\dot{\bmx}) =\varkappa^2.
\]
Moreover, one also has that
\[
\bmx' = \tau' \dot{\bmx}. 
\]

\medskip
Making use of the above observations, one can compute
\begin{eqnarray*}
&& \langle \bmbeta, \dot{\bmx} \rangle = \langle \bmbeta,
   \frac{1}{\tau'} \bmx' \rangle \\
  && \phantom{\langle \bmbeta, \dot{\bmx} \rangle} = \langle \alpha
     \bmx^{\prime \flat},   \frac{1}{\tau'} \bmx'\rangle\\
  &&  \phantom{\langle \bmbeta, \dot{\bmx} \rangle} =
     \frac{\alpha}{\tau'} \tilde{\bmg}(\bmx',\bmx') =
     \frac{\alpha}{\tau'} . 
\end{eqnarray*}
In addition, observe that from
\[
1= \tilde{\bmg} (\bmx',\bmx') =\tau^{\prime 2} \tilde{\bmg}(\dot{\bmx},\dot{\bmx}),
\]
it follows that
\begin{eqnarray*}
&& \Theta^2 =\tau^{\prime 2} \Theta^2
   \tilde{\bmg}(\dot{\bmx},\dot{\bmx}) \\
&& \phantom{\Theta^2 } = \tau^{\prime 2} \bmg(\dot{\bmx},\dot{\bmx})
   =\varkappa^2\tau^{\prime 2}.  
\end{eqnarray*}
Accordingly, choosing the positive root, one concludes that
\[
\Theta = \varkappa \tau^\prime.
\]
It then follows from the equation $\dot{\Theta}_{\star} =\langle \bmbeta_{\star},
\dot{\bmx}_{\star}\rangle \Theta_{\star}$ that
\[
\dot{\Theta}_{\star} = \varkappa \alpha_{\star}.
\]
Finally, it is observed that
\begin{eqnarray*}
&& \tilde{\bmg}^\sharp(\bmbeta,\bmbeta) = \alpha^2
\tilde{\bmg}^\sharp(\bmx^{\prime\flat}, \bmx^{\prime\flat})\\
&& \phantom{ \tilde{\bmg}^\sharp(\bmbeta_\star,\bmbeta_\star)} =
\alpha^2 \tilde{\bmg}(\bmx',\bmx') =\alpha^2.
\end{eqnarray*}
Thus, the second constraint takes the form
\begin{eqnarray*}
  && \Theta_\star \ddot{\Theta}_\star = \frac{\varkappa^2}{2}\alpha^2_\star\\
  && \phantom{ \Theta_\star \ddot{\Theta}_\star} = \frac{1}{2}\dot{\Theta}_{\star}^2,
\end{eqnarray*}
so that, in fact, one has
\begin{equation}
    \ddot{\Theta}_{\star} =\frac{\dot{\Theta}_\star^2}{2\Theta_\star}.
    \label{Relation-ddotTheta-dotTheta}
\end{equation}

\medskip
One then has the following lemma:

\begin{lemma}[conformal factor data in terms of reparametrisation data]
  \label{Lemma:InitialData}
  One has that
  \[
\Theta_\star = \varkappa \tau^\prime_\star, \qquad \dot{\Theta}_\star =\varkappa 
\alpha_\star, \qquad \ddot{\Theta}_\star = \frac{\varkappa \alpha^2_\star}{2\tau^\prime_\star}.
\]
\end{lemma}

\begin{remark}
{\em This last result shows the connection between the initial data specifying the reparametrisation of standard geodesics as conformal geodesics and shows how this reparametrisation naturally specifies a conformal extension.}
\end{remark}

\subsection{Geodesics starting from a Cauchy hypersurface}
\label{Section: Geodesics starting from a Cauchy hypersurface}
In order to gain some intuition, we first consider a reparametrisation of the congruence of geodesics \eqref{TimelikeGeodesics} where the data of the reparametrisation \eqref{GeneralReparametrisation1}-\eqref{GeneralReparametrisation2} are prescribed on the hypersurface $\mathcal{S}_\star$ described by the condition $t=0$ ---see Figure \ref{Timelike-geodesics-to-infinity}.

\begin{figure}[t]
\begin{center}
\includegraphics[width=0.3\textwidth]{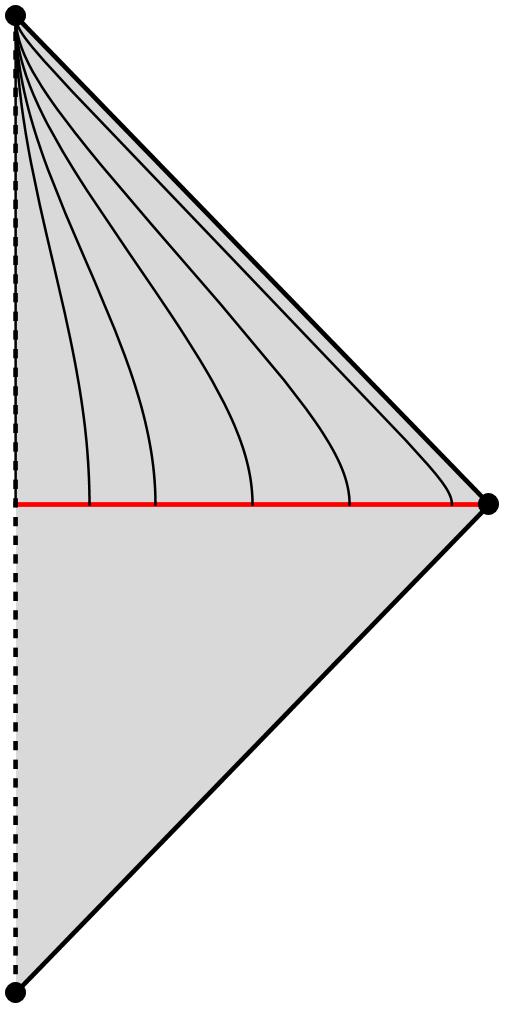}
\put(-50,190){$\scri^+$}
\put(-50,55){$\scri^-$}
\put(-125,260){$i^+$}
\put(-125,-10){$i^-$}
\put(-90,115){${\mathcal{S}_\star}$}
\end{center}
\caption{Penrose diagram of the Minkowski spacetime depicting members of the family of geodesics \eqref{TimelikeGeodesics} with tangent $\bmpartial_t$ starting from the standard Cauchy hypersurface $\mathcal{S}_\star$ given by the condition $t=0$.}
\label{Timelike-geodesics-to-infinity}
\end{figure}

\begin{figure}[t]
\begin{center}
\includegraphics[width=0.5\textwidth]{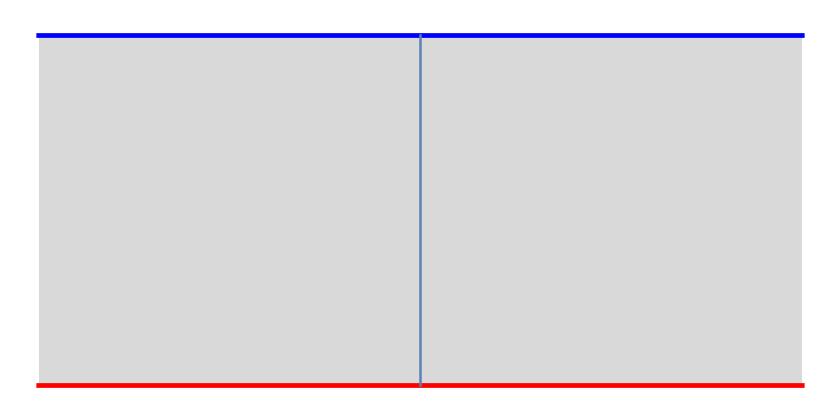}
\put(-150,-4
){${\mathcal{S}_\star}$}
\put(-150,104
){${\mathcal{S}^+}$}
\put(-100,50){$\gamma$}
\end{center}
\caption{The simple blow up of timelike infinity using geodesics starting from the  Cauchy hypersurface $\mathcal{S}_\star$ given by the condition $t=0$. The hypersurface $\mathcal{S}^+$ given by the condition $\tau=1$ is given by the (future) endpoints of the curves and, thus it describes, in some sense, (future) timelike infinity.}
\label{Simple-blow-up-cauchy-surface}
\end{figure}

\medskip
A particular case of
\eqref{GeneralReparametrisation1}-\eqref{GeneralReparametrisation2} is
given by
\[
s(\tau) =\frac{\tau}{1-\tau}, \qquad \alpha(\tau) = 2(\tau-1),
\]
which corresponds to the initial data 
\begin{equation}
    s_{\star} = 0, \qquad \tau_{\star}=0, \qquad \tau'_{\star}=1, \qquad \alpha_{\star} = -2,
    \label{Initial-data-Cauchy-hypersurface}
\end{equation}
and the choice $\varkappa=1$. Notice also that this solution satisfies $\alpha(1) =0$.
Moreover, $\tau \rightarrow 1$ as $s \rightarrow \infty$. The resulting family of conformal geodesics is given by
\[
x(\tau) =\left( \frac{\tau}{1-\tau}, r_\star,\theta_\star,\varphi_\star  \right), \qquad \tau\in[0,1]. 
\]
Thus, the reparametrisation has effectively introduced an end-point at $\tau=1$ representing the point at (future) infinity of the curve ---see Figure \ref{Simple-blow-up-cauchy-surface}. The parameter $\tau$ can be used as new time coordinate:
\[
t = \frac{\tau}{1-\tau} \Longleftrightarrow \tau =\frac{t}{1+t}.
\]
In particular, one has that
\[
\mathbf{d}t = \frac{1}{(1-\tau)^2}\mathbf{d}\tau. 
\]
It follows that the tangent to the conformal geodesics is given by
\begin{eqnarray*}
  && \bmx ' = \bmpartial_{s} = \bmpartial_{t}, \qquad \dot{\bmx} = \dot{t} \bmpartial_t = \frac{1}{(1-\tau)^2}\bmpartial_t, \\
  && \phantom{\bmx ' = \bmpartial_{s} = \bmpartial_{t}, \qquad \dot{\bmx}= \dot{t} \bmpartial_t } = \bmpartial_\tau,
\end{eqnarray*}
while the associated 1-form
\begin{eqnarray*}
  && \bmbeta = -\frac{2}{1+t}\mathbf{d}t \\
  && \phantom{\bmbeta} = -\frac{2}{1-\tau}\mathbf{d}\tau. 
\end{eqnarray*}

\subsubsection{The associated conformal factor}
One can now make use of Proposition \ref{Lemma:ConformalGaussianGauge} and Lemma \ref{Lemma:InitialData}
to construct a conformal factor associated to the congruence. This conformal factor is fixed by requiring $\tau$ to be the $\bmg$-proper time
along the curves ---that is, by setting $\varkappa=1$ in
Proposition \ref{Lemma:ConformalGaussianGauge}, so that
\[
\bmg(\dot{\bmx},\dot{\bmx}) = \Theta^2 \tilde{\bmg}(\dot{\bmx},\dot{\bmx}). 
\]
More precisely, one has that
\[
\Theta(\tau) = \Theta_\star + \dot{\Theta}_\star (\tau-\tau_\star) +
\frac{1}{2}\ddot{\Theta}_\star (\tau-\tau_\star)^2,
\]
with the coefficients $\Theta_\star$, $\dot{\Theta}_\star$ and
$\ddot{\Theta}_\star$ subject to the constraints
\[
\dot{\Theta}_\star =\langle \bmbeta_\star, \dot{\bmx}_\star \rangle
   \Theta_\star, \qquad
\Theta_\star \ddot{\Theta}_\star = \frac{1}{2}\tilde{\bmeta}^\sharp(\bmbeta_\star,\bmbeta_\star).
\]
From the initial data \eqref{Initial-data-Cauchy-hypersurface}, one has $\tau_\star=0$. Then, making use of Lemma \ref{Lemma:InitialData} with $\varkappa=1$, we have
$$
\Theta_\star =1, \qquad \dot{\Theta}_\star = -2, \qquad \ddot{\Theta}_\star =2.
$$
Therefore, in this case, no attempt is made to a conformal compactification of the spatial sections since $\Theta_\star =1$. The conformal factor corresponding to this initial data is then given by
\[
\Theta = (1-\tau)^2.
\]

\begin{remark}
{\em Observe that both $\Theta$ and $\mathbf{d}\Theta$ vanish for
  $\tau=1$. This property characterises timelike infinity in 
  Minkowski-like spacetimes. That is, one has that 
\[
\Theta|_{\tau=1}=0, \qquad \mathbf{d}\Theta|_{\tau=1}=0.
\]  
  }
\end{remark}

We also observe that
\[
\bmeta = \mathbf{d}\tau\otimes \mathbf{d}\tau - (1-\tau)^4 \bmdelta.
\]
From the above, it is clear that the constant $\tau$ slices are \emph{conformally flat}. The simple blow-up of timelike infinity introduced by this reparametrisation is illustrated in Figure~\ref{Simple-blow-up-cauchy-surface}.

\begin{remark}
{\em The discussion of the precedent paragraphs provide a \emph{blow-up} of timelike infinity $i^+$ in the sense that it sets in one-to-one correspondence points of the initial hypersurface $\mathcal{S}_\star$ with the endpoints of the congruence of orbits of the Killing vector $\bmpartial_t$. The latter set is denoted by $\mathcal{S}_+$ and corresponds to the locus of points for which $\tau=1$ ---see Figure \ref{Simple-blow-up-cauchy-surface}. Recall that in Penrose's traditional conformal extension of the Minkowski spacetime, the endpoint of all timelike geodesics is given by the point $i^+$. }
\end{remark}

\begin{remark}
{\em The construction in the previous paragraphs is based on a specific congruence of timelike geidesics ---namely, the integral curves of the timetranslation Killing vector $\bmpartial_t$. However, a similar construction can be implemented, \emph{mutatis mutandi}, with any other congurence of timelike geodesics. As it will be discussed in the next subsection, this observation is crucial to understand why our construction does not lead to a projective compactification.}    
\end{remark}

\subsubsection{Relation to projective extensions}
\label{Subsubsection:CauchyProjectiveExtension}

A natural question concerns the relation of the constructions in the precedent subsections to projective extensions. 

\medskip
Setting, for convenience,  $\varsigma\equiv 1-\tau$ so that $\varsigma\rightarrow 0$ as $t\rightarrow \infty$ one can write the Minkowski line element as
\[
\tilde\bmeta = \frac{1}{\varsigma^4}\mathbf{d}\varsigma\otimes \mathbf{d}\varsigma - \delta_{\alpha\beta}\mathbf{d}x^\alpha\otimes \mathbf{d}x^\beta.
\]
The only non-zero Christoffel symbol is then given by
\[
\tilde\Gamma_\varsigma{}^\varsigma{}_\varsigma =-\frac{2}{\varsigma}.
\]
We make use of the transformation formula \eqref{ProjectiveConnection} with $\varpi = \mathbf{d}(\ln \varsigma )$ to define a projectively related connection $\grave{\bmnabla}$. One has that 
\[
\grave\Gamma_{\varsigma}{}^{\varsigma}{}_{\varsigma} = 0,
\]
however, one also has that 
\[
\grave\Gamma_\varsigma{}^\alpha{}_\beta = \frac{1}{\varsigma}\delta_\beta{}^\upalpha
\]
which is clearly singular at $\varsigma=0$. Consequently, the boundary defining function $\varsigma=1-\tau$ associated to the constructions of the previous subsections cannot be used, in a natural manner, to obtain a regular projective extension. 

\begin{remark}
{\em This problem arises form the fact that our construction hinges on a specific family of timelike geodesics ---namely those for which $x^\alpha(t)/t \rightarrow 0$ for $t\rightarrow \infty$. In the remainder of this discussion, we slightly abuse notation and redefine the boundary defining function by setting $\varsigma = 1/t$. With this convention, a generic timelike geodesic in the Minkowski spacetime is of the form
\[
\big( t, x^\alpha(t) \big) = ( t, x^\alpha_\star + v^\alpha_\star t) 
\]
with $x^\alpha_\star$, $v^\alpha_\star$ constants. Thus, as $\varsigma \rightarrow 0$ one has that $x^\alpha(t)$ behaves like $v^\alpha_\star/\varsigma$,  which generically diverges. In other words, although $\varsigma\rightarrow 0$ along a generic timelike geodesic, the spatial coordinates $x^\alpha(\varsigma)$ diverge, so that the coordinates $(\varsigma,x^\alpha(\varsigma))$ do not provide a finite coordinate representation of its endpoint at the boundary.
}
\end{remark}


In order to construct a regular projective extension, consider new spatial coordinates $\underline{y}=(y^\alpha)$ given by
\begin{equation}
    y^\alpha \equiv \varsigma x^\alpha=\frac{x^\alpha}{t}.
\end{equation}
Along a generic timelike geodesic, these coordinates satisfy
\begin{equation}
    y^\alpha=\frac{x^\alpha_\star}{t}+v^\alpha_\star
    \rightarrow v^\alpha_\star,
    \qquad t\rightarrow\infty,
\end{equation}
so that the limiting values of $y^\alpha$ encode the asymptotic spatial velocity of the geodesic. In these coordinates, a generic timelike geodesic is of the form
\[
\big( t, y^\alpha(t) \big) = \Big( t, \frac{x^\alpha_\star}{t} + v^\alpha_\star \Big).
\]
The Minkowski line element can be written in terms of $(\varsigma,y^\alpha)$ as
\begin{equation}
\tilde \bmeta = \frac{1}{\varsigma^4}(1-|y|^2) \mathbf{d}\varsigma \otimes\mathbf{d}\varsigma + \frac{1}{\varsigma^3}\delta_{\alpha\beta}y^\beta ( \mathbf{d}\varsigma \otimes \mathbf{d}y^\alpha + \mathbf{d}y^\alpha \otimes \mathbf{d}\varsigma) -\frac{1}{\varsigma^2}\delta_{\alpha\beta}\mathbf{d}y^\alpha \otimes \mathbf{d}y^\beta, 
\label{MinkowskiRescaledSpatialCoordinates}
\end{equation}
where $|y|^2\equiv \delta_{\alpha\beta}y^\alpha y^\beta$. Observe that in order to keep the Lorentzian signature of the metric one requires that $|y|^2 <1$. This expression of the Minkowski line element suggests that, in order to obtain a projective compactification, one should introduce rescaled spatial coordinates which remain finite along generic timelike geodesics as $t\rightarrow\infty$.

\medskip
The non-vanishing Christoffel symbols of the line element \eqref{MinkowskiRescaledSpatialCoordinates} are 
\[
\tilde{\Gamma}_\varsigma{}^\varsigma{}_\varsigma =-\frac{1}{\varsigma}, \qquad \tilde{\Gamma}_\varsigma{}^\alpha{}_\beta = \tilde{\Gamma}_\beta{}^\alpha{}_\varsigma = -\frac{1}{\varsigma}\delta_\beta{}^\alpha.  
\]
Proceeding as before and using the transformation formula \eqref{ProjectiveConnection} with $\varpi = \mathbf{d}(\ln \varsigma)$, one finds that 
\[
\grave{\Gamma}_\mu{}^\lambda{}_\nu =0.
\]
In other words, the extended connection $\grave{\bmnabla}$ is the flat affine connection which trivially extends smoothly through the boundary defined by the condition $\varsigma=0$. Expressed in terms of the coordinates $(\varsigma,y^\alpha)$, a generic timelike geodesic $(\varsigma,y^\alpha(\varsigma))$ has the limit $(0,v^\alpha_\star)$.

\section{Conformal Gaussian systems form hyperboloids and \texorpdfstring{$\mathscr{I}$}{I}-fixing} 
\label{Section:Conformal Gaussian Systems}

In this section, we discuss the construction of conformal Gaussian systems for the Minkowski spacetime from a reference hyperboloidal hypersurface. These constructions will give rise not only to a representation of future timelike infinity in which this subset of the conformal boundary is described by an extended set (i.e. no longer a point) but also will implement $\mathscr{I}$-fixing ---i.e. a gauge for which future null infinity is described as the locus of points for which a spatial coordinate vanishes.

\subsection{Hyperboloids in the Minkowski spacetime}
We begin by making a brief review of the construction of hyperboloidal hypersurfaces (hyperboloids for short) in the Minkowski spacetime by means of height functions. The
approach follows the discussion in \cite{Zen08b,Zen11a}. 

\medskip
Given the standard Minkowski metric in spherical coordinates,
we introduce a new time coordinate $\psi$ via a height function
$f(r)$. Namely, we set
\begin{equation}
\psi \equiv t -f(r).
\label{Definition:HeightFunction}
\end{equation}
A natural choice for the height function is
\begin{equation}
    f(r) =\sqrt{\fraka^2 + r^2}, 
    \label{Choice-height-function}
\end{equation}
where $\fraka$ is a positive constant ---it is, however, convenient to keep the following discussion as general as possible. Observe that the new coordinate $\psi$ preserves the time translation Killing vector in the sense that $\bmpartial_\psi = \bmpartial_t$. Notice, however, that
\[
\mathbf{d}t = \mathbf{d}\psi + f'(r)\mathbf{d}r,
\]
where in this section ${}^\prime$ denotes differentiation with respect
to the argument of the function. One thus has that
\begin{equation}
    \tilde{\bmeta} =\mathbf{d}\psi\otimes \mathbf{d}\psi + f'(r) (\mathbf{d}\psi \otimes \mathbf{d}r + \mathbf{d}r\otimes \mathbf{d}\psi) -\big(1-f^{\prime 2}(r)\big)\mathbf{d}r \otimes \mathbf{d}r -r^2 \slashed \bmsigma.
    \label{physical-metric-psi-r}
\end{equation}
From the above, the intrinsic metric of the hypersurfaces of constant $\psi$ is given by
\[
\tilde{\bmh} = -\big(1-f^{\prime 2}(r)\big)\mathbf{d}r \otimes \mathbf{d}r -r^2 \slashed \bmsigma.
\]

\begin{remark}
{\em The hypersurfaces of constant coordinate $\psi$ as given by equation \eqref{Definition:HeightFunction} with the choice \eqref{Choice-height-function} for $f(r)$ are examples of constant mean curvature (CMC) hyperboloids. Indeed, one can readily verify the relation 
\[
(t-\psi)^2 -r^2 =\fraka^2.
\]
It can be verified that the hypersurfaces described by the above condition are umbilical ---that is, their extrinsic curvature is proportional to the intrinsic metric. Accordingly, the extrinsic curvature is described entirely in terms of the mean curvature
\[
\tilde{K} = \frac{3}{\fraka}. 
\]
For fixed value of the constant $\fraka$, different values of $\psi$ gives a foliation of CMC hyperboloids. Some examples of hyperboloidal hypersurfaces for various values of $\fraka$ and $\psi$ are shown in Figure \ref{Constant-psi-hypersurfaces}. The CMC hyperboloid with $\fraka=1$ and $\psi=0$ is known as the \emph{standard unit hyperboloid.} }
\end{remark}

\begin{remark}
{\em Recall that the (negative definite) metric of the $3$-dimensional hyperbolic space $\mathbb{H}^3$ with curvature radius $\fraka>0$ can be written as
\begin{equation*}
    \tilde{\bmh}_{\mathbb{H}^3} = - \fraka^2 \big(\mathbf{d} \varrho \otimes \mathbf{d} \varrho +  \sinh^{2} \varrho \slashed \bmsigma\big).
\end{equation*}
The metric $\tilde{\bmh}$ is isometric to $\tilde{\bmh}_{\mathbb{H}^3}$ if there exists a coordinate transformation $r \mapsto \varrho(r)$ such that
\begin{equation}
    \sqrt{1-f^{\prime 2}(r)} \mathbf{d}r= \fraka \mathbf{d}\varrho, \qquad r = \fraka \sinh \varrho.
    \label{Change-coordinates-r-to-varrho}
\end{equation}
In other words, the metric $\tilde{\bmh}$ is hyperbolic if the height function $f(r)$ satisfies
\begin{equation}
    f'(r)^2 = \frac{r^2}{\fraka^2+r^2}.
    \label{Height-function-condition}
\end{equation}
This is indeed the case for $f(r)$ in equation \eqref{Choice-height-function} and motivates this specific choice of height function.}
\end{remark}

\begin{figure}[t]
\begin{center}
\includegraphics[width=0.3\textwidth]{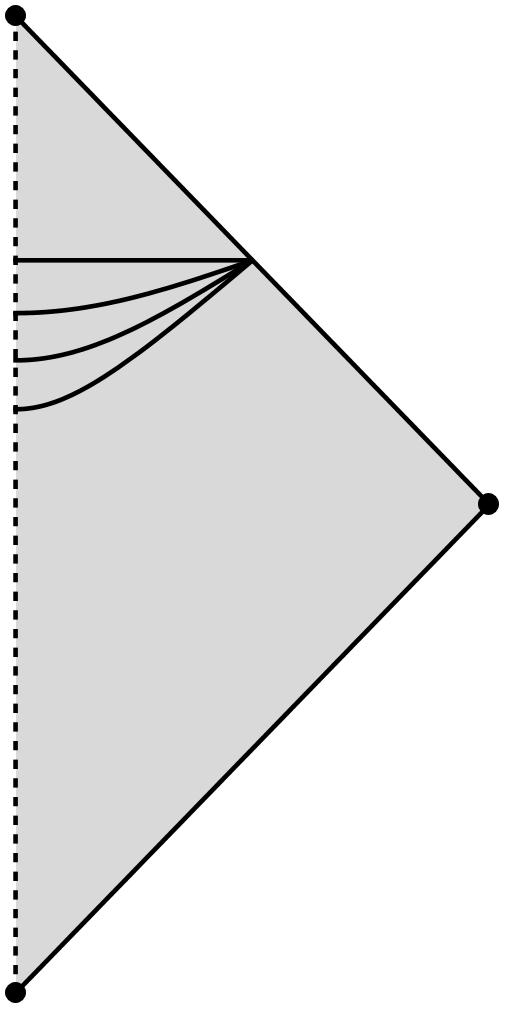}
\put(-50,190){$\scri^+$}
\put(-50,55){$\scri^-$}
\put(-125,260){$i^+$}
\put(-125,-10){$i^-$}
\put(-160, 145){$\fraka^2 =\frac{1}{10}$}
\put(-160, 160){$\fraka^2 =\frac{1}{4}$}
\put(-160, 175){$\fraka^2 =\frac{1}{2}$}
\put(-160, 190){$\fraka^2 =1$}
\end{center}
\caption{Penrose diagram of the Minkowski spacetime describing the location of CMC hyperboloids for $\psi=0$ and various choices of the constant $\fraka$. This diagram is qualitatively accurate.}
\label{Constant-psi-hypersurfaces}
\end{figure}

\begin{figure}[t]
\begin{center}
\includegraphics[width=0.3\textwidth]{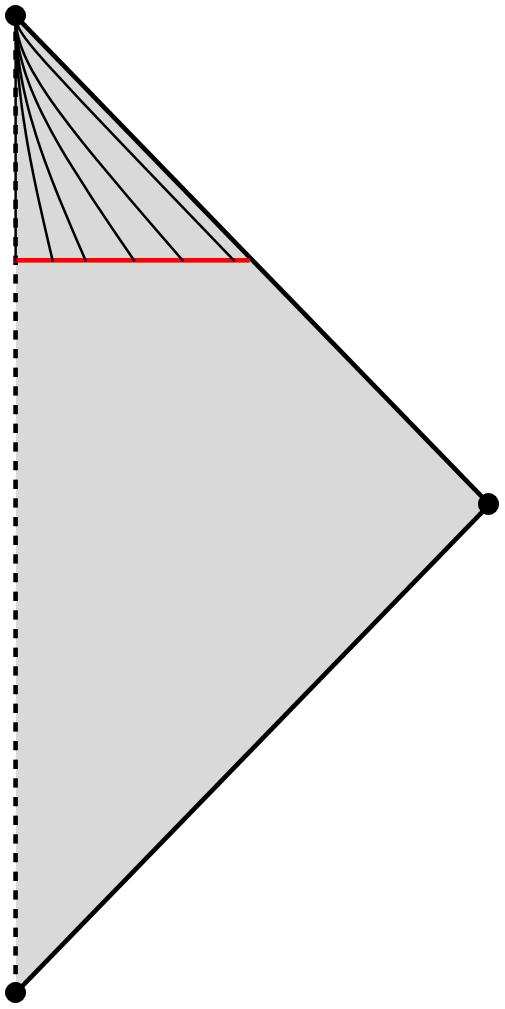}
\put(-50,190){$\scri^+$}
\put(-50,55){$\scri^-$}
\put(-125,260){$i^+$}
\put(-125,-10){$i^-$}
\put(-100,175){${\mathcal{H}_\star}$}
\end{center}
\caption{Penrose diagram of the geodesics starting from a hyperboloidal hypersurface.}
\label{Reparameterisation-from-hyperboloid}
\end{figure}

\subsubsection{Compactification into the interior of \texorpdfstring{$\mathbb{S}^2$}{S2} ---the Poincar\'e ball model} 
\label{Section:Compactification-2-Sphere}
Since $\tilde{\bmh} \equiv \tilde{\bmh}_{{\mathbb{H}^3}}$ with $r \mapsto \varrho{(r)}$ and $f(r)$ satisfying \eqref{Height-function-condition}, we can work directly with $\tilde{\bmh}_{{\mathbb{H}^3}}$ to show that ${\mathbb{H}^3}$ can be compactified into the interior of the unit ball $\mathbb{B}^{3} \subset \mathbb{R}^{3}$ with a conformal boundary $\mathbb{S}^2$. To this end, we write the flat Euclidean metric as
\[
\bmdelta = - \big(\mathbf{d}\chi \otimes \mathbf{d}\chi + \chi^2 \slashed\bmsigma\big).
\]
Now, consider a conformal factor $\varpi$ such that
\begin{eqnarray*}
    && \bmdelta = \varpi^2 \tilde{\bmh}_{{\mathbb{H}^3}}, \\
    && \phantom{\bmdelta } = - \fraka^2 \big(\varpi^{2}  \mathbf{d} \varrho \otimes \mathbf{d} \varrho +\varpi^{2} \sinh^{2} \varrho \slashed \bmsigma\big).
\end{eqnarray*}
From the above, it is clear that we are looking for a conformal factor $\varpi$ and a coordinate $\chi$ such that
\begin{equation}
    \mathbf{d} \chi = \fraka \varpi \mathbf{d} \varrho, \qquad \chi = \fraka \varpi \sinh{\varrho}.
    \label{change-of-coordinates-varrho-to-chi}
\end{equation}
It follows from these conditions, after some manipulations, that 
\begin{equation*}
    \chi'(\varrho)= \frac{\chi(\varrho)}{\sinh{\varrho}},
\end{equation*}
which can be solved to give
\begin{equation*}
    \chi(\varrho)= \tanh{\frac{\varrho}{2}}
\end{equation*}
where the integration constant was fixed so that $ \chi \to 1$ as $\varrho \to \infty$.
 With this, we also have
\begin{equation*}
    \varpi = \frac{1}{\fraka+\fraka \cosh{\varrho}},
\end{equation*}
which, in turn, shows that
$\varpi \to 0$ as  $\varrho \to \infty$.
 From \eqref{Change-coordinates-r-to-varrho} and \eqref{change-of-coordinates-varrho-to-chi}, we have
\begin{equation}
    r = \frac{2 \fraka \chi}{1-\chi^2} \implies \mathbf{d} r =
\frac{2\fraka(1+\chi^2)}{(1-\chi^2)^2}\mathbf{d}\chi. 
\label{RToChi}
\end{equation}
We can also write $\varpi$ and $f'(r)$ as a function of $\chi$ so as to obtain 
\[
\varpi =\frac{(1-\chi^2)}{2\fraka}, \qquad f'(\chi)= \frac{2 \chi}{1+\chi^2}.
\]  

\medskip
The expressions in the previous paragraphs can be used to express the
Minkowski line element in terms of the coordinates $(\psi,\chi)$. One
finds that
\begin{equation}
\tilde{\bmeta} = \mathbf{d}\psi \otimes \mathbf{d} \psi
+\frac{4\fraka\chi}{(1-\chi^2)^2}(\mathbf{d}\psi \otimes
\mathbf{d}\chi + \mathbf{d}\chi \otimes \mathbf{d}\psi ) -\varpi^{-2}
\big( \mathbf{d}\chi\otimes \mathbf{d}\chi + \chi^2 \slashed\bmsigma
\big).
 \label{PhysicalMinkowski:psi-chiCoordinates}
\end{equation}

\begin{remark}
{\em The conformal rescaling in the previous paragraphs relating the metric $\tilde{\bmh}_{\mathbb{H}^3}$ to the flat metric in the unit ball $\mathbb{B}^3$ is essentially the so-called Poincar\'e ball model of hyperbolic geometry --- see \cite[Theorem~3.7(c)]{Lee2018}.}    
\end{remark}

\subsubsection{An alternative spatial compactification ---Zenginoglu's gauge}
\label{Section:Another-spatial-compactification}
In this subsection, we discuss an alternative spatial compactification which does not involve a conformal rescaling of the intrinsic metric to a portion of the Euclidean space. This approach is the one used most often used in Numerical Relativity ---see e.g.  \cite{Zen08b,Zen11a}. In what follows we will refer to it as \emph{Zenginoglu's gauge}.

\medskip
In the following, consider a new coordinate $\rho$ via the relation
\[
r= \frac{\rho}{\Omega(\rho)},
\]
where $\Omega$ is a boundary defining function. A simple choice is given by
\begin{equation}
    \Omega = 1-\rho,
    \label{Choice-of-Omega}
\end{equation}
so that $\mathbf{d}\Omega =-1\neq 0$ at $\rho=1$. Nevertheless, as before, it is better to keep the discussion general. For convenience, define
\[
F(\rho) \equiv f \left( \frac{\rho}{\Omega(\rho)}\right), \qquad
F'(\rho) \equiv \frac{\mathrm{d}F}{\mathrm{d}\rho}.
\]
Observing that
\[
\frac{\mathrm{d}}{\mathrm{d}\rho} \left(
  \frac{\rho}{\Omega(\rho)}\right) = \frac{\Omega -\rho \Omega'}{\Omega^2},
\]
one concludes that
\[
\mathrm{d}r = \left(  \frac{\Omega -\rho \Omega'}{\Omega^2} \right) \mathrm{d}\rho,
\]
and, moreover, using the chain rule, that
\[
f'(r) =\left( \frac{\Omega^2}{\Omega-\rho \Omega'} \right)F'(\rho).
\]

\medskip
Making use of the above expressions, the line Minkowski line element takes
the form
\begin{equation}
\tilde{\bmeta} = \mathbf{d}\psi\otimes \mathbf{d}\psi + F'(\rho) (
\mathbf{d}\psi \otimes \mathbf{d}\rho + \mathbf{d}\rho \otimes
\mathbf{d}\psi) -\left( \left(\frac{\Omega -\rho \Omega'}{\Omega^2} \right)^2-F^{\prime 2}(\rho) \right) \mathbf{d}\rho\otimes \mathbf{d}\rho
-\left( \frac{\rho}{\Omega}\right)^2 \slashed\bmsigma.
\label{Minkowski:Zenginoglu}
\end{equation}
Writing  
$$
G(\rho) \equiv \frac{(\Omega-\rho \Omega')^2}{\Omega^4},
$$
then the contravariant (i.e. inverse) metric $\tilde{\bmeta}^{\sharp}$ can be written as
\begin{equation*}
    \tilde{\bmeta}^{\sharp} = \left(1-\frac{F'^2(\rho)}{G(\rho)} \right) \bmpartial_{\psi}\otimes \bmpartial_{\psi} + \frac{F'(\rho)}{G(\rho)} (\bmpartial_{\psi} \otimes \bmpartial_{\rho} + \bmpartial_{\rho} \otimes \bmpartial_{\psi}) - \frac{1}{G(\rho)} \bmpartial_{\rho} \otimes \bmpartial_{\rho} - \left( \frac{\Omega}{\rho} \right)^2 \slashed{\bmsigma}^{\sharp}.
\end{equation*}
Both $\tilde{\bmeta}$ and $\tilde{\bmeta}^{\sharp}$ can be explicitly computed for the choice of $\Omega$ and $f(r)$ in equations \eqref{Choice-of-Omega} and \eqref{Choice-height-function}.

\medskip
In the present case, the intrinsic metric of the hypersurfaces of constant coordinate $\psi$ is given by
\[
\tilde{\bmh} = -\left( \left(\frac{\Omega -\rho \Omega'}{\Omega^2} \right)^2-F^{\prime 2}(\rho) \right) \mathbf{d}\rho\otimes \mathbf{d}\rho
-\left( \frac{\rho}{\Omega}\right)^2 \slashed\bmsigma.
\]
Following the discussion in \cite{Zen08b}, we introduce a conformal rescaling of $\tilde{\bmh}$ via
\begin{equation}
    \bmh = \Omega^2 \tilde{\bmh}
    \label{Another-spatial-compactification}
\end{equation}
so that 
\begin{equation*}
    \bmh = -\left( \frac{(\Omega -\rho \Omega')^2}{\Omega^2}
    -\Omega^2 F^{\prime 2}(\rho) \right)
    \mathbf{d}\rho\otimes \mathbf{d}\rho
    - \rho^2 \slashed\bmsigma.
\end{equation*}
If $\bmh$ is to admit a regular extension to $\rho=1$, the potentially divergent term
in the radial part of the line element must cancel out. More precisely, it
is sufficient to require
\begin{equation*}
    (\Omega-\rho\Omega')^2-\Omega^4F'(\rho)^2
    =O(\Omega^2) \qquad \text{as } \rho\to1.
\end{equation*}
In this way the coefficient of $\mathbf{d}\rho \otimes \mathbf{d} \rho$ approaches a finite, non-zero value at $\rho=1$. For the
explicit choices of $\Omega$ and $f(r)$ in
\eqref{Choice-of-Omega} and \eqref{Choice-height-function}, this
condition is indeed satisfied.

\begin{remark}
{\em Note that $\tilde{\bmh}=
\tilde{\bmh}_{\mathbb{H}^3}$. Therefore, by construction,
\[
\bmh=\Omega^2\tilde{\bmh}_{\mathbb{H}^3},
\]
so that $\bmh$ is conformal to the hyperbolic metric. Provided that the
radial coefficient of $\bmh$ admits a finite and non-vanishing limit
at $\rho=1$, the metric $\bmh$ extends regularly to the compactified
hypersurface, whose boundary is represented by the $2$-sphere
$\rho=1$. Thus, the transformation
\eqref{Another-spatial-compactification} gives a conformal
compactification of the hyperboloids with constant coordinate $\psi$.}
\end{remark}

\subsection{A gauge system using the Poincar\'e ball model of hyperbolic space}
\label{Subsection:BeltramiKlein}

In Subsection \ref{Section:Compactification-2-Sphere}, it has been shown that the Minkowski metric can be written in the form \eqref{PhysicalMinkowski:psi-chiCoordinates}. This form of the Minkowski line element clearly brings to the forefront the relation of the hypersurfaces of constant coordinate $\psi$ with the Poincar\'e ball model of hyperbolic geometry. For this line element the conformal factor $\varpi$ plays the role of a boundary defining function which vanishes for $\chi=1$. Intuitively, the associated boundary corresponds to sections of future null infinity. 

\medskip
Proceeding by analogy to Subsection \ref{Section: Geodesics starting from a Cauchy hypersurface}, in order to implement a conformal Gaussian system we consider the reparametrisation
\begin{equation}
    \psi = \frac{\tau}{1-\tau}, \qquad \mathbf{d}\psi = \frac{\mathbf{d}\tau}{(1-\tau)^2}. 
    \label{Reparameterisation-psi-tau}
\end{equation}
The above reparametrisation follows from the choice 
\begin{equation}
\tau_\star=0, \quad \psi_\star=0, \quad \alpha_\star=-2, \quad \tau'_\star =1,
\label{PoincareReparametrisationParameters}
\end{equation}
in formula \eqref{GeneralReparametrisation1} with $\psi$ playing the role of the parameter $s$. This renders a Minkowski line element of the form
\begin{eqnarray}
&& \tilde{\bmeta} = \frac{1}{(1-\tau)^4}\mathbf{d}\tau \otimes \mathbf{d} \tau \nonumber \\
&& \hspace{2cm}+\frac{4\fraka\chi}{(1-\tau)^2(1-\chi^2)^2}(\mathbf{d}\tau \otimes
\mathbf{d}\chi+
 \mathbf{d}\chi \otimes \mathbf{d}\tau )\nonumber \\
&& \hspace{4cm} -\frac{4\fraka^2}{(1-\chi^2)^2}
\big( \mathbf{d}\chi\otimes \mathbf{d}\chi + \chi^2 \slashed\bmsigma
\big).
\label{Minkowski:BeltramiKlein}
\end{eqnarray}
\begin{remark}
    {\em Note that in terms of $(\tau,\chi,\theta,\varphi)$ coordinates, the congruence of conformal geodesics corresponding to $\gamma_{(r_\star,\theta_\star,\varphi_\star)}$ introduced in \eqref{TimelikeGeodesics} would be written as
    \begin{equation*}
        x(\tau) = \left( \frac{\tau}{1-\tau}, \chi_{\star}, \theta_\star,\varphi_\star \right).
    \end{equation*}
    Then, the tangent vector $\dot{\bmx}(\tau)$ and the 1-form $\bmbeta(\tau)$ satisfying the conformal geodesic equations are given by
    \begin{eqnarray*}
        && \dot{\bmx}(\tau) = \bmpartial_{\tau}, \qquad \bmbeta(\tau) = - \frac{2}{1-\tau} \mathbf{d} \tau - \frac{8 \fraka \chi (1-\tau)}{1-\chi^2} \mathbf{d} \chi.
    \end{eqnarray*}}
\end{remark}
Next, we consider the rescaling
\[
\bmeta = \Theta^2 \tilde\bmeta,
\]
with
\begin{eqnarray}
&& \Theta \equiv \vartheta \varpi \nonumber \\
&& \phantom{\vartheta} = \frac{(1-\chi^2)}{2\fraka}(1-\tau)^2, \label{ConformalFactor:Beltrami-KleinCGSystem}
\end{eqnarray}
where $\vartheta\equiv (1-\tau)^2$. This conformal factor follows from Proposition \ref{Proposition:Canonical-conformal-factor} consistent with the choice of reparametrisation as given by \eqref{PoincareReparametrisationParameters} and Lemma \ref{Lemma:InitialData}. One readily finds that 
\begin{equation}
\bmeta = \frac{(1-\chi^2)^2}{4\fraka^2} \mathbf{d}\tau \otimes \mathbf{d}\tau + \frac{1}{\fraka}\chi (1-\tau)^2 \big( \mathbf{d}\tau \otimes \mathbf{d}\chi + \mathbf{d}\chi\otimes\mathbf{d}\tau \big)-(1-\tau)^4 \big(\mathbf{d}\chi \otimes \mathbf{d}\chi +\chi^2\slashed\bmsigma  \big).
\label{ConformalMinkowski:B-K}
\end{equation}
In particular, observe that 
\[
\bmeta(\bmpartial_\tau,\bmpartial_\tau) = \varpi^2.
\]
Moreover, it can be seen that
\begin{eqnarray}
    && \bmeta^\sharp = \frac{4\fraka^2}{(1+\chi^2)^2}\bmpartial_\tau\otimes\bmpartial_\tau + \frac{4\fraka \chi }{(1-\tau)^2(1+\chi^2)^2} \big(\bmpartial_\tau \otimes\bmpartial_\chi + \bmpartial_\chi \otimes \bmpartial_\tau  \big) \nonumber \\
    && \phantom{\bmeta^\sharp = \frac{4\fraka^2}{(1+\chi^2)^2}\bmpartial_\tau\otimes\bmpartial_\tau + \qquad } - \frac{(1-\chi^2)^2}{(1+\chi^2)^2(1-\tau)^4}\bmpartial_\chi\otimes\bmpartial_\chi - \frac{1}{\chi^2(1-\tau)^4}\slashed\bmsigma^\sharp.  \qquad 
    \label{EtaSharp}
\end{eqnarray}

\medskip
In the following, we consider the line element given by \eqref{ConformalMinkowski:B-K}. We regard this as a metric defined on the domain
\[
\mathscr{D}\equiv \{ (\tau,\chi,\theta^{\mathcal{A}})\in [0,1]\times [0,1] \times \mathbb{S}^2 \}
\]
which, we argue, provides a conformal extension of the future domain of dependence $D^+(\tilde{\mathcal{H}}_\star)$ of the standard unit hyperboloid $\tilde{\mathcal{H}}_\star$ in the Minkowski spacetime given by
\[
\tilde{\mathcal{H}}_\star \equiv \{p\in \mathbb{R}^{1,3} \;|\; \psi(p) =0 \}
\]
---i.e. the standard unit hyperboloid; see Figure \ref{Timelike-geodesic-diagram}. An appealing feature of this representation is that the spatial sections (i.e. hypersurfaces of constant $\tau$) are conformally flat. In the conformal representation of $D^+(\tilde{\mathcal{H}}_\star)$ given by $\mathscr{D}$, the conformal factor $\Theta$ as given by \eqref{ConformalFactor:Beltrami-KleinCGSystem} vanishes at points described by the conditions $\chi=1$ and/or $\tau=1$. Accordingly, we define the sets
\begin{eqnarray*}
&& \mathscr{I}^+ \equiv \{ p\in \mathscr{D}\;|\; \chi(p)=1 \},\\
&& \mathcal{H}^+ \equiv \{p \in \mathscr{D}\; |\; \tau(p)=1 \}, \\
&& \mathcal{C}^+\equiv \{ p \in \mathscr{D} \;|\; \chi(p)=1, \; \tau(p)=1\},
\end{eqnarray*}
corresponding, respectively to \emph{future null infinity}, the \emph{hyperboloid at future timelike infinity} and the \emph{final cut of future null infinity} ---see Figure \ref{Scri-fixing-region}. Strictly speaking what is described above as future null infinity is, in fact, the portion of it which is the future of the set 
\[
\mathcal{H}_\star \equiv \{ p\in \mathcal{D} \;|\; \tau(p)=0 \}
\]
---i.e. the conformal extension of $\tilde{\mathcal{H}}_\star$ which is sometimes called the \emph{unphysical unit hyperboloid}.

\begin{figure}[t]
\begin{center}
\includegraphics[width=0.3\textwidth]{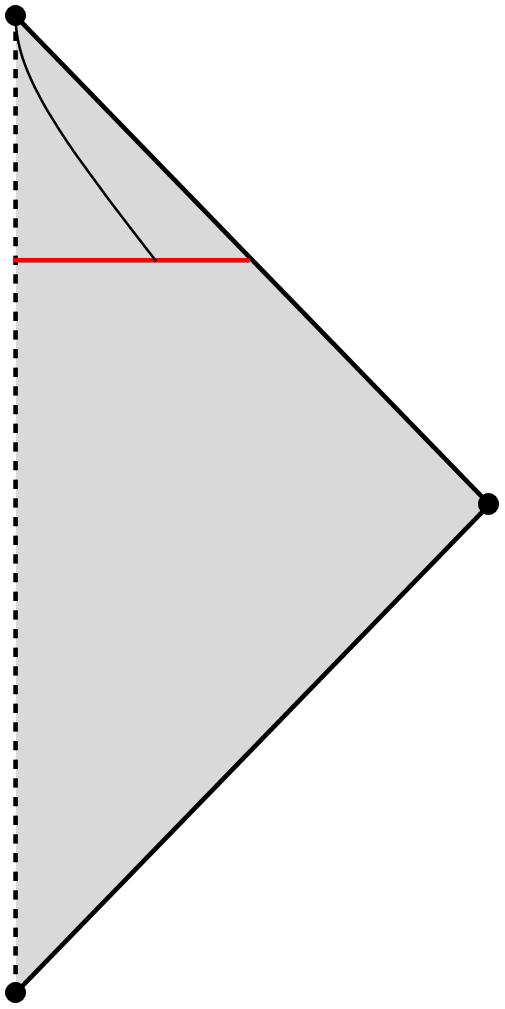}
\put(-50,190){$\scri^+$}
\put(-50,55){$\scri^-$}
\put(-125,260){$i^+$}
\put(-125,-10){$i^-$}
\put(-100,175){${\mathcal{H}_\star}$}
\end{center}
\caption{Penrose diagram of the future domain of dependence $D^+(\tilde{\mathcal{H}})$ of the standard unit hyperboloid $\tilde{\mathcal{H}}$ (red horizontal line). The diagram also depicts one of the curves belonging to the congruence of conformal geodesic ---namely the one starting at the location $\chi_\star =1/2$. This diagram is quantitatively accurate.}
\label{Timelike-geodesic-diagram}
\end{figure}

\begin{remark}
   { \em Observe that, as one would expect, one has that 
    \[
\bmeta(\bmpartial_\tau,\bmpartial_\tau)\big|_{\mathscr{I}^+}=0,
    \]
    so that the congruence of conformal geodesics becomes null at future null infinity. Furthermore, the pullback of the conformal metric to
    $\mathscr{I}^+$ is given by 
    \[
    \bmh_{\mathscr{I}^+}
    \equiv -(1-\tau)^4\slashed{\bmsigma}.
    \]
    Hence, $\bmh_{\mathscr{I}^+}$ is a rank-2 degenerate metric whose
    Kernel is generated by the vector $\bmpartial_\tau$. Consequently, as usual, $\mathscr{I}^+$ carries a \emph{Carrollian structure} ---see e.g. \cite{DuGibHor14,CiaLeiMarPetro19,CiaJai25}. At $\mathcal{C}^+$, the induced metric vanishes identically ---i.e. one has that $\bmh_{\mathscr{I}^+}\big|_{\tau=1}=0$, so that the Carrollian structure degenerates completely at the corner. }
\end{remark}

\begin{remark}
{ \em   One also has that 
    \[
    \bmeta|_{\tau=1} =\frac{1}{4\fraka^2}(1-\chi^2)^2 \mathbf{d}\tau \otimes \mathbf{d}\tau.
    \]
    This coincides with the discussion of Section \ref{Section: Geodesics starting from a Cauchy hypersurface}.
    }     
\end{remark}

\begin{remark}
{\em The construction in this section provides a realisation of what is known in the literature as \emph{$\mathscr{I}$-fixing} ---i.e. a gauge in which the location of $\mathscr{I}^+$ is given in terms of the locus of points for which a certain \emph{spatial coordinate} acquires a specific value. The latter, in turn, gives rise to a boundary defining function. Crucially, as discussed in Section \ref{Section:NonlinearPerturbations}, this feature of the gauge also holds for nonlinear perturbations of the Minkowski spacetime ---thus, enabling a framework for the study of perturbations of the the Minkowski spacetime by analytic and/or numerical methods with $\mathscr{I}$-fixing. }    
\end{remark}

\begin{figure}[t]
\begin{center}
\includegraphics[width=0.5\textwidth]{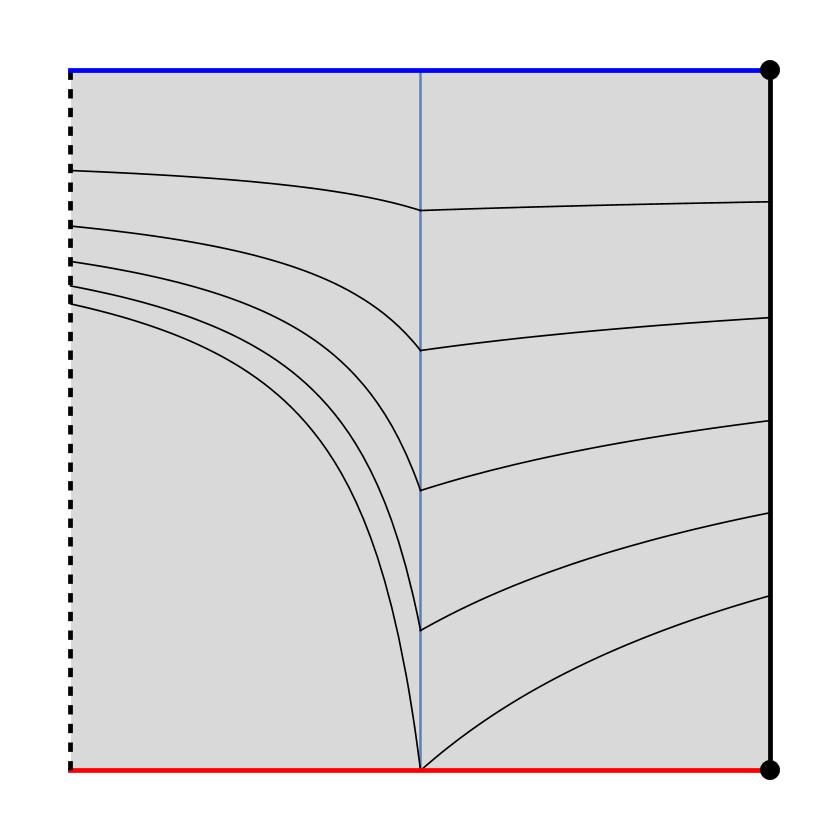}
\put(-12,100){$\scri^+$}
\put(-110,5
){${\mathcal{H}_\star}$}
\put(-110,202
){${\mathcal{H}^+}$}
\put(-12,202
){${\mathcal{C}^+}$}
\put(-12,5
){${\mathcal{C}_\star}$}
\end{center}
\caption{Diagram of the domain $\mathscr{D}$ with timelike infinity described in terms of the hyperboloid at timelike infinity $\mathcal{H}^+$. The light blue line corresponds to the timelike geodesic in Figure \ref{Timelike-geodesic-diagram} and the light black lines are the light cones at various events along this curve.}
\label{Scri-fixing-region}
\end{figure}

\subsection{The conformal Gaussian coordinates in Zenginoglu's gauge}
The construction described in the previous section can be readily adapted to the hyperboloidal hypersurfaces given in Subsection \ref{Section:Another-spatial-compactification}. As already mentioned, this approach is favoured in numerical studies. 

\medskip
Using the reparameterisation \eqref{Reparameterisation-psi-tau}, we can write $\tilde{\bmeta}$ as given by the line element \eqref{Minkowski:Zenginoglu}, in the form
\begin{eqnarray}
        && \tilde{\bmeta} = \frac{1}{(1-\tau)^4} \mathbf{d}\tau\otimes \mathbf{d}\tau + \frac{F'(\rho)}{(1-\tau)^2} ( \mathbf{d}\tau \otimes \mathbf{d}\rho + \mathbf{d}\rho \otimes \mathbf{d}\tau) \nonumber \\
        && \phantom{\tilde{\bmeta} = \frac{1}{(1-\tau)^4} \mathbf{d}\tau\otimes \mathbf{d}\tau  \frac{F'(\rho)}{(1)^2}}-\left( \left(\frac{\Omega -\rho \Omega'}{\Omega^2} \right)^2-F^{\prime 2}(\rho) \right) \mathbf{d}\rho\otimes \mathbf{d}\rho -\left( \frac{\rho}{\Omega}\right)^2 \slashed\bmsigma.
        \label{Physical-metric-conformal-Gaussian-gauge}
\end{eqnarray}
We can then introduce the conformal metric 
\begin{equation*}
    \bar{\bmeta} = \bar{\Theta}^2 \tilde{\bmeta},
\end{equation*}
where 
\begin{equation*}
    \bar{\Theta} \equiv \vartheta \Omega = (1-\tau)^2 (1-\rho), 
\end{equation*}
where we have used $\vartheta \equiv  (1-\tau)^2$. From the above, we have
\begin{multline*}
\bar{\bmeta} = \Omega^2 \mathbf{d}\tau\otimes \mathbf{d}\tau + F'(\rho) (1-\tau)^2 \Omega^2 ( \mathbf{d}\tau \otimes \mathbf{d}\rho + \mathbf{d}\rho \otimes \mathbf{d}\tau) \\
-\left( \frac{(\Omega -\rho \Omega')^2}{\Omega^2}-\Omega^2 F^{\prime 2}(\rho) \right) (1-\tau)^4 \mathbf{d}\rho\otimes \mathbf{d}\rho - \rho^2 (1-\tau)^4 \slashed\bmsigma.
\end{multline*}
Substituting the explicit expressions for $\Omega$ and $F'(\rho)$ in terms of the radial coordinate $\rho$ one finds that 
\begin{eqnarray}
&& \bar{\bmeta} =  (1-\rho)^2 \mathbf{d}\tau\otimes \mathbf{d}\tau + \frac{\rho (1-\tau)^2}{(1-\rho) \sqrt{\fraka^2 + \tfrac{\rho^2}{(1-\rho)^2}}} (
\mathbf{d}\tau \otimes \mathbf{d}\rho + \mathbf{d}\rho \otimes
\mathbf{d}\tau) \nonumber\\
&& \hspace{5cm}- \frac{\fraka^2 (1-\tau)^4}{\fraka^2 (1-\rho)^2+\rho^2} \mathbf{d}\rho\otimes \mathbf{d}\rho
- \rho^2 (1-\tau)^4 \slashed\bmsigma,
\label{Metric:ZenginogluGauge}
\end{eqnarray}
while the contravariant metric $\bar{\bmeta}^{\sharp}$ can be written as 
\begin{multline}
\bar{\bmeta}^{\sharp} = \frac{\fraka^2}{\fraka^2 (1-\rho)^2 +\rho^2} \bmpartial_{\tau} \otimes \bmpartial_{\tau} - \frac{\rho}{(1-\rho) (1-\tau)^2 \sqrt{\fraka^2 + \tfrac{\rho^2}{(1-\rho)^2}}} \left( \bmpartial_{\tau} \otimes \bmpartial_{\rho} + \bmpartial_{\rho} \otimes \bmpartial_{\tau}  \right) \\
- \frac{(1-\rho)^2}{(1-\tau)^4} \bmpartial_{\rho} \otimes \bmpartial_{\rho} - \frac{1}{\rho^2 (1-\tau)^4} \slashed\bmsigma^{\sharp}. \label{EtaBarSharp}
\end{multline}

\begin{remark}
{\em In this case, the metric on the constant $\tau$ hypersurface is given by --- compare with equation (9) in \cite{Zen08b}
\begin{equation*}
    \bar{\bmh} = - \frac{\fraka^2 (1-\tau)^4}{\fraka^2 (1-\rho)^2+\rho^2} \mathbf{d}\rho\otimes \mathbf{d}\rho
- \rho^2 (1-\tau)^4 \slashed\bmsigma.
\end{equation*}}
\end{remark}

Similar to the discussion in the previous section, one can see that 
\[
\bar{\bmeta}(\bmpartial_\tau,\bmpartial_\tau)|_{ \rho =1} =0,
\]
indicating that future null infinity corresponds to the locus of points for which $\rho=1$. Also, the metric on the hypersurfaces of constant coordinate $\rho$ is given by
\begin{equation*}
    \bar{\bmq} = \Omega^2 \mathbf{d}\tau\otimes \mathbf{d}\tau - \rho^2 (1-\tau)^4 \slashed\bmsigma,
\end{equation*}
which, for $\tau=1$, reduces to 
\[
\bar{\bmq}|_{\tau=1} = \Omega^2 \mathbf{d}\tau \otimes \mathbf{d}\tau.
\]
Therefore, the metric completely degenerates at the corner $\tau=1$, $\rho=1$. 

\begin{remark}
{\em Proceeding by analogy to the discussion in the previous subsection, one can define \emph{mutatis mutandi}, a domain $\bar{\mathscr{D}}$ via
\[
\bar{\mathscr{D}}\equiv \{ (\tau,\rho,\theta^{\mathcal{A}})\in [0,1]\times [0,1]\times \mathbb{S}^2  \} .
\]
which provides an alternative conformal representation of $D^+(\tilde{\mathcal{H}}_\star)$. The sets $\mathscr{I}^+$, $\mathcal{H}^+$ and $\mathcal{C}^+$ in the conformal boundary are defined analogously ---see Figure \ref{DomainEstimates}.
}
\end{remark}

\begin{remark}
{\em The construction in this subsection provides an additional example of $\mathscr{I}$-fixing gauge. }
\end{remark}

\begin{figure}[t]
\begin{center}
\includegraphics[width=0.5\textwidth]{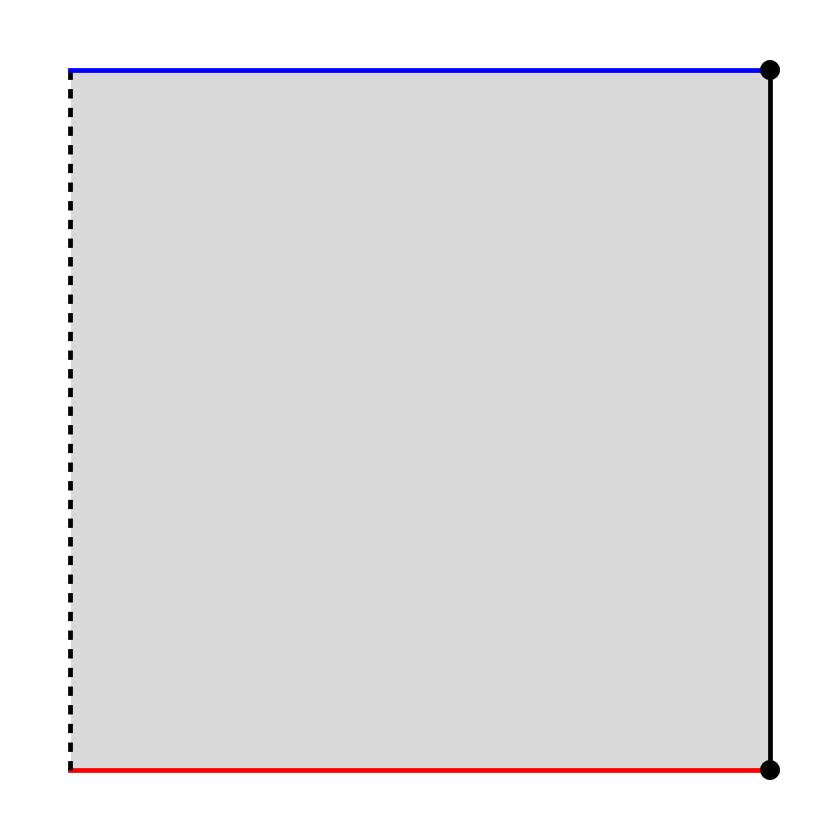}
\put(-12,100){$\scri^+$}
\put(-110,5
){${\mathcal{H}_\star}$}
\put(-110,202
){${\mathcal{H}^+}$}
\put(-110,105){$\bar{\mathscr{D}}$}
\put(-12,202
){${\mathcal{C}^+}$}
\put(-12,5
){${\mathcal{C}_\star}$}
\end{center}
\caption{Diagram of the $\scri$-fixing region with timelike infinity blown-up to an hyperboloid $\mathcal{H}^+$.}
\label{DomainEstimates}
\end{figure}

\section{Hyperbolic equations}
\label{Section:MasslessFields}
In this section, we explore the way the gauge constructions in Section \ref{Section:Conformal Gaussian Systems} interact with test massless fields. We are interested, in particular, on the behaviour of the solutions to these equations towards the hyperboloid at timelike infinity $\mathcal{H}^+$ and, specially, the corner $\mathcal{C}^+$.  

\subsection{The scalar wave equation}
We start by looking at the wave equation
\begin{equation}
\square \phi =0,
\label{ScalarWaveEquation}
\end{equation}
on the conformally rescaled domains $\mathscr{D}$ and $\bar{\mathscr{D}}$. 

\begin{remark}
{\em Strictly speaking, rather than the equation \eqref{ScalarWaveEquation}, one should consider the so-called conformally invariant wave equation 
\[
\square \phi - \frac{1}{6} R\phi =0.
\]
Here, for simplicity of the presentation consider \eqref{ScalarWaveEquation}. The general qualitative behaviour of the solutions is similar in both cases.
}
\end{remark}

\subsubsection{Poincar\'e ball gauge}
A calculation readily shows that for the line element \eqref{ConformalMinkowski:B-K} as discussed in Subsection \ref{Subsection:BeltramiKlein}, the scalar wave equation \eqref{ScalarWaveEquation} takes the form
\begin{eqnarray*}
&& \frac{4\fraka^2(1-\tau)^4}{(1+\chi^2)^2}\partial^2_\tau\phi + \frac{8\fraka (1-\tau)^2\chi}{(1+\chi^2)^2}\partial_\tau\partial_\chi \phi - \frac{(1-\chi^2)^2}{(1+\chi^2)^2}\partial^2_\chi \phi \\
&& \hspace{1cm} + \frac{4\fraka}{(1+\chi^2)^2}(1-\tau)^2\big( 3+\chi^2 + 6\fraka(\tau-1)(1+\chi^2) \big)\partial_\tau \phi \\
&& \hspace{2cm} -\frac{1}{\chi (1+\chi^2)^3}\Big( 2\big(1-(4+8\fraka(\tau-1))\big)\chi^2+\big(1+8\fraka(1-\tau) \big)\chi^4 + 2\chi^6\Big)\partial_\chi\phi\\
&& \hspace{3cm} -\frac{1}{\chi^2}\slashed\Delta\phi=0,  
\end{eqnarray*}
where $\slashed\Delta\phi$ denotes the Laplacian of the round metric  $\slashed \bmsigma$ of $\mathbb{S}^2$.

\medskip
\noindent
\textbf{Behaviour at $\mathcal{H}^+$. }
At $\tau=1$, the above equation degenerates into the elliptic equation
\[
\frac{(1-\chi^2)^2}{(1+\chi^2)^2}\partial^2_\chi \phi + \frac{2(1-\chi^2)}{\chi(1+\chi^2)^3}(2\chi^4 +3\chi^2-1)\partial_\chi \phi + \frac{1}{\chi^2}\slashed \Delta\phi=0.
\]

\begin{remark}
\label{Remark:EllipticEqn}
    {\em The above elliptic equation suggests that the values of the scalar filed $\phi$ are determined by its value at the corner $\mathcal{C}^+$. Notice, however, that the equation degenerates at $\chi=1$ ---thus, questions regarding its solvability cannot be readily addressed by means of standrad methods of the theory of elliptic equations. A detailed study of this equation will be provided elsewhere.  }
\end{remark}

\medskip
\noindent
\textbf{Behaviour at $\mathscr{I}^+$.} On $\mathscr{I}^+$ (i.e. $\chi=1$), one obtains the transport equation
\begin{eqnarray}
&& \fraka^2 (1-\tau)^4 \partial^2_\tau\phi +\frac{\fraka}{2}(1-\tau)^2 \big( 4-12\fraka (1-\tau)\big)\partial_\tau \phi -\slashed\Delta \phi \nonumber \\
&& \hspace{5cm}= 4\fraka(1-\tau)\partial_\chi\phi -2\fraka (1-\tau)^2\partial_\tau\partial_\chi\phi. \label{WaveEqn:Scri}
\end{eqnarray}
The above equation can be integrated if the transverse derivative $\partial_\chi\phi$ is known at $\chi=1$. Insight into the solutions to this equations can be obtained using a decomposition of $\phi$ into spherical harmonics. To this end, recall that
\[
\slashed\Delta Y_{\ell m}=-\ell(\ell+1) Y_{\ell m}. 
\]
In this spirit, consider the ordinary differential equation
\begin{equation}
(1-\tau)^4 \ddot\varphi_{\ell} +\frac{1}{2\fraka}(1-\tau)^2 \big( 4-12\fraka (1-\tau)\big)\dot\varphi_{\ell} +\frac{1}{\fraka^2}\ell(\ell+1) \varphi_{\ell} =0,
\label{ODE:Poincare}
\end{equation}
corresponding to the decomposition in spherical harmonics of the homogeneous part of equation \eqref{WaveEqn:Scri}. For this ordinary differential equation, one has the following lemma describing behaviour of its solutions as $\tau\rightarrow 1$:

\begin{lemma}
\label{Lemma:BehaviourCPlus}
The solutions of equation \eqref{ODE:Poincare} satisfy, for $\ell=0$
\[
\varphi_{0}(\tau) = \mathfrak{b}_0 + \mathfrak{c}_0e^{-\frac{2}{\fraka(1-\tau)}}\left(1-\tau\right)^{-4}\left(1 + O(1-\tau)\right),
\]
with $\mathfrak{b}_0$ and $\mathfrak{c}_0$ two constants. In particular, one has that
\[
\varphi_0 (\tau) \rightarrow  \mathfrak{b}_0, \quad \mbox{for} \quad \tau\rightarrow 1.
\]
For $\ell\geq 1$, one obtains two linearly independent solutions satisfying
\[
\varphi_{\ell, \pm}(\tau) \sim (1-\tau)^{-2} e^{-\frac{1}{\fraka(1-\tau)}} \exp\left(\pm i\left(\frac{\kappa_\ell}{\fraka(1-\tau)} - \frac{2}{\kappa_\ell} \log(1-\tau)\right)\right)
\]
with 
\[
\kappa_\ell\equiv  \sqrt{\ell(\ell +1) - 1}.
\]
Accordingly, one has that 
\[
\varphi_{\ell,\pm}(\tau)\rightarrow 0 , \quad \mbox{for} \quad \tau\rightarrow 1.
\]
\end{lemma}

The calculations leading to the above result are given in Appendix \ref{Solution-transport-equation}. Note that $\sim$ here implies asymptotic equivalence as $\tau \to 1$, i.e. the ratio of LHS and RHS tends to $1$. 

\subsubsection{Zenginolgu's gauge}
For the metric $\bar\bmeta$ as given by the line element \eqref{Metric:ZenginogluGauge}, one finds that the scalar wave equation $\square_{\bar{\bmeta}} \phi=0$ takes the form 
\begin{eqnarray}
&&  \frac{\fraka^2}{R^2} \partial^{2}_{\tau}\phi +  \frac{2 \rho}{R (1-\tau)^2}\partial_{\tau}\partial_{\rho} \phi - \frac{(1-\rho)^2}{(1-\tau)^4} \partial^{2}_{\rho} \phi \nonumber \\
&& \hspace{2cm}+ \frac{2}{R \rho (1-\tau)^4} (3R \rho - R -2 \rho^2 - 2 R \rho^2+ 2 \rho^2 \tau)\partial_{\rho}\phi\nonumber \\
 && \hspace{3cm} + \frac{2\rho^2 + \fraka^2 (3-6R + \rho (2\rho-5)+ 6 R \tau)}{R^3 (1-\tau)^2} \partial_{\tau}\phi  \nonumber \\
 && \hspace{4cm}-\frac{1}{\rho^2 (1-\tau)^4} \slashed\Delta\phi =0,
\end{eqnarray}
with
\begin{equation}
    R\equiv \sqrt{\fraka^2(1-\rho)^2+\rho^2} = \sqrt{\fraka^2\Omega^2+\rho^2}.
    \label{Definition-R}
\end{equation}
It is clear from the above expressions that the restriction of $\square_{\bar{\bmeta}} \phi$ to $\tau=1$ gives
\begin{equation*}
    \rho (1-\rho)^2 \partial_{\rho}^2 \phi+ (1-\rho) (2-4\rho)\partial_{\rho} \phi+ \frac{1}{\rho} \slashed\Delta\phi =0.
\end{equation*}
This, as in the case of the Poincar\'e ball gauge, is a degenerate elliptic equation ---see Remark \ref{Remark:EllipticEqn}. On $\rho=1$, since $R|_{\rho=1}=1$ and $\Omega|_{\rho=1}=0$, we obtain
\begin{equation*}
    \fraka^2 (1-\tau)^4 \partial_{\tau}^2 \phi + (1-\tau)^2 (2-6\fraka^2 (1-\tau)) \partial_{\tau} \phi - \slashed\Delta\phi= 4 (1-\tau) \partial_{\rho} \phi-2 (1-\tau)^2 \partial_{\tau} \partial_{\rho} \phi.
\end{equation*}
As in the Poincaré ball gauge, one can study the homogeneous part of this equation ---assuming that the right-hand side are known source terms. This analysis provides insight into the behaviour of the solutions to the wave equation towards the corner $\mathcal{C}^+$. The general qualitative behaviour is similar to that described in Lemma \ref{Lemma:BehaviourCPlus}. Further details can be found in Appendix \ref{Solution-transport-equation}.

\subsection{Linearised gravity}

In this section, we consider linearised gravitational perturbations of Minkowski spacetime, described by the massless spin-2 equation, and derive the corresponding equations in the conformal gauge constructed above. Our objective here is simply to obtain the explicit form of the resulting system. A detailed analysis of the differential equations, including their behaviour at the conformal boundary and the associated evolution problem, lies beyond the scope of the present article and will be addressed elsewhere.

\medskip
It is convenient to formulate the linearised Einstein equations in the \emph{2-spinor formalism} where gravitational perturbations are represented by a totally symmetric massless spin-2 field satisfying the massless spin-2 equation
\[
\nabla^A{}_{A'}\phi_{ABCD}=0,
\]
where $\phi_{ABCD}=\phi_{(ABCD)}$. The latter, in turn, implies the wave equation
\[
\square \phi_{ABCD} + 6 \Psi^{PQ}{}_{(AB} \phi_{CD)PQ}- 12\Lambda \phi_{ABCD}=0. 
\]
In the case under consideration, one has that $\Psi_{ABCD}=0$.
Thus, $\phi_{ABCD}$ satisfies the simpler equation
\[
\square \phi_{ABCD}-12\Lambda \phi_{ABCD}=0.
\]
\begin{remark}
    {\em Observe that in the above equation the D'Alembertian operator acts over a valence 4 symmetric spinor. Thus, the expanded operator contains derivatives of the connection coefficients.}
\end{remark}
The massless spin-2 field $\phi_{ABCD}$ is described by 5 complex components $\phi_0$, $\cdots$, $\phi_4$ each with a specific spin-weight. More precisely, for the component $\phi_k$ one has that $\mbox{\$}\phi_k =2-k$.
Using the suite {\tt xAct} for the {\tt Wolfram Programming Language} \cite{xAct}, one readily finds, using the so-called \emph{compact NP notation} \cite{PenRin84},  that 
\begin{eqnarray*}
&&  \square\phi_{0} - 8 \uprho \eth \phi_{1}- 4  \uprho
\uprho ' \phi_{0} -12 \Lambda \phi_{0} = 0, \\
&&  \square \phi_{1}  - 6 \uprho \eth \phi_{2} - 2 \uprho ' \eth' \phi_{0}- 10  \uprho
\uprho '\phi_{1} -12 \Lambda \phi_{1}  = 0\\
&&  \square \phi_{2}  - 4 \uprho \eth \phi_{3} - 4 \uprho ' \eth' \phi_{1}- 12  \uprho
\uprho '\phi_{2} -12 \Lambda  \phi_{2}  = 0, \\
&&  \square \phi_{3}  - 2 \uprho \eth \phi_{4} - 6 \uprho ' \eth' \phi_{2}- 10  \uprho 
\uprho '\phi_{3} -12 \Lambda  \phi_{3} = 0, \\
&&  \square \phi_{4}  - 8 \uprho ' \eth' \phi_{3}- 4  \uprho 
\uprho '\phi_{4}-12 \Lambda  \phi_{4}  = 0.
\end{eqnarray*}

\begin{remark}
{\em The spin connection coefficients $\uprho$ and $\uprho'$ are, for the class of NP frames we will be considering, are non-vanishing. Accordingly, the evolution of all the components is coupled and any analysis of the equation requires to consider the equations simultaneously.
 It is well-known that there are gauges where this coupling is not present.}
\end{remark}

\subsubsection{Expressions in the Poincar\'e ball gauge}
In the following, we will consider the NP frame $\{ \bml,\, \bmn,\, \bmm,\, \bar{\bmm}  \}$ for the metric $\bmeta$ given by 
\begin{eqnarray*}
 && \bml =\frac{\sqrt{2}\fraka(1-\chi)^2}{1-\chi^4}\bmpartial_\tau +\frac{1-\chi^2}{\sqrt{2}(1-\tau)^2(1+\chi^2)}\bmpartial_\chi, \\
        && \bmn =\frac{\sqrt{2}\fraka(1+\chi)^2}{1-\chi^4}\bmpartial_\tau -\frac{1-\chi^2}{\sqrt{2}(1-\tau)^2(1+\chi^2)}\bmpartial_\chi, \\
&&\bmm = \frac{1}{\chi(1-\tau)^2}\bmpartial_+, \\
&& \bar\bmm = \frac{1}{\chi(1-\tau)^2}\bmpartial_-,    
\end{eqnarray*}
where $\{\bmpartial_+,\, \bmpartial_-\}$ span $T\mathbb{S}^2$. For this frame, the spin-connection coefficients $\uprho$ and $\uprho'$ are given, in the Poincar\'e ball gauge, by
\[
\uprho =\frac{(\chi-1)\Big( 1+ \chi \big(\chi+2 +4\fraka (\tau-1) \big)\Big)}{\sqrt{2}\chi(1+\chi)(1+\chi^2)(1-\tau)^2}, \qquad \uprho'= \frac{(\chi+1)\Big( 1+ \chi \big(\chi-2 -4\fraka (\tau-1) \big)\Big)}{\sqrt{2}\chi(\chi-1)\chi(1+\chi^2)(1-\tau)^2},
\]
and the Ricci scalar $\Lambda$ is given by
\[
\Lambda = \frac{3-\chi^2 (2+\chi^2)-4 \fraka (3+ \chi^2) (1-\tau) + 12 \fraka^2 (1+\chi^2) (1-\tau)^2}{2 (1+\chi^2)^3 (1-\tau)^4}.
\]
Observe that while $\uprho'$ is singular at $\chi=1$, the product $\uprho \uprho'$ is regular. Details of the NP-frame used to compute the connection coefficients can be found in Appendix \ref{Appendix:NPComputations}. Note that $\square \phi$ denotes the D'Alembertian operator acting on a scalar with spin-weight $s$. In particular, we can write
\[
\square = \blacksquare - \frac{1}{\chi^2(1-\tau)^4}(\eth_{\mathbb{S}^2} \eth_{\mathbb{S}^2}'+ \eth_{\mathbb{S}^2}'\eth_{\mathbb{S}^2}), 
\]
with $\blacksquare$ involving only $\partial_\tau$ and $\partial_\chi$ derivatives (or $\partial_\rho$ for the other spatial compactification),
\begin{eqnarray*}
    && \blacksquare = \frac{4 \fraka^2}{(1+\chi^2)^2} \partial^{2}_{\tau} + \frac{8 \fraka \chi}{(1+\chi^2)^2 (1-\tau)^2} \partial_{\tau} \partial_{\chi} - \frac{(1-\chi^2)^2}{(1+\chi^2)^2 (1-\tau)^4} \partial^{2}_{\chi} \\
    && \phantom{\blacksquare = } + \frac{4 \fraka (3+\chi^2 - 6 \fraka (1+\chi^2)(1-\tau))}{(1+\chi^2)^3 (1-\tau)^2} \partial_{\tau} \\
    && \phantom{\blacksquare = } - \frac{2 (1+2 \chi^6 + \chi^2 (8 \fraka (1-\tau)-4)+ \chi^4 (1+8 \fraka (1-\tau) ))}{\chi (1+\chi^2)^3 (1-\tau)^4} \partial_{\chi},
\end{eqnarray*}
while $\eth_{\mathbb{S}^2}$ and $\eth_{\mathbb{S}^2}'$ are the $\eth$ and $\eth'$ operators of the round sphere so that $\eth_{\mathbb{S}^2} \eth_{\mathbb{S}^2}'+ \eth_{\mathbb{S}^2}'\eth_{\mathbb{S}^2}$ 
corresponds to the Laplacian on spin-weighted scalars.

\begin{remark}
{\em The wave equations describing the evolution of the various components of the massless spin-2 field are mildly singular at $\mathscr{I}^+$. This type of mild singular behaviour is amenable to an analytic treatment by means of the tools b-calculus ---see e.g. \cite{HinVas20,HinVas26}. This will be discussed elsewhere.}    
\end{remark}

\begin{remark}
{\em Expressions with  qualitatively analogous structural properties are obtained for the Zenginoglu gauge.}
\end{remark}


\section{Projective extensions and Ti}
\label{Section: projective-extnesion-Ti}
In this section, we discuss the relation of the constructions of the hyperboloid at infinity $\mathcal{H}^+$ in Section \ref{Section:Conformal Gaussian Systems} to projective extensions describing timelike infinity and, in particular, the so-called extended boundary Ti. 

\subsection{Projective extensions}
We start by carrying out an analogue of Subsection \ref{Subsubsection:CauchyProjectiveExtension} to the hyperboloidal setting.

\medskip
Letting, as before, $\varsigma \equiv 1-\tau$, then the physical Minkowski metric $\tilde{\bmeta}$ with line element given by \eqref{Minkowski:BeltramiKlein} takes the form
\[
\tilde{\bmeta} =\frac{1}{\varsigma^4}\mathbf{d}\varsigma \otimes \mathbf{d}\varsigma - \frac{2\fraka\chi}{\varsigma^2(1-\chi^2)^2}\big( \mathbf{d}\varsigma\otimes\mathbf{d}\chi +\mathbf{d}\chi\otimes\mathbf{d}\varsigma  \big) -\frac{4\fraka^2}{(1-\chi^2)^2}\big(\mathbf{d}\chi \otimes \mathbf{d}\chi -\chi^2 \slashed\bmsigma \big).
\]
A simple computation readily shows that one cannot obtain a projective extension of this line element with boundary defining function given by $\varsigma$. As in Subsection \ref{Subsubsection:CauchyProjectiveExtension}, this difficulty can be understood by looking at the behaviour of generic timelike geodesics. In spherical coordinates the curves can be described by $x(s) =(s, r(s), \vartheta_\star^\mathcal{A})$ with $r(s)= v_\star s + r_\star$ with $v_\star$, $r_\star$ constants. For $v_\star\neq 0$, as $s\rightarrow\infty$ (i.e. $\varsigma\rightarrow 0$), one gets $r(s)\rightarrow \infty$ and, moreover, using formula \eqref{RToChi} that $\chi(s)\rightarrow 1$. Thus, all the geodesics with $v_\star\neq0$ have their future endpoint at the corner $\mathcal{C}^+$ ---see Figure \ref{Figure:TimelikeGeodesics}.

\begin{figure}[t]
\begin{center}
\includegraphics[width=0.5\textwidth]{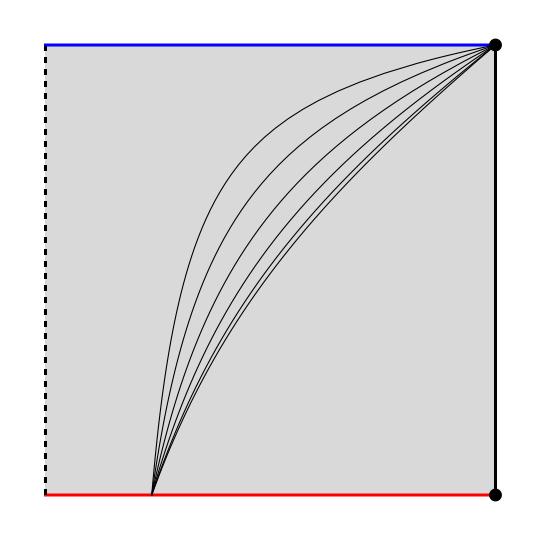}
\put(-12,98){$\scri^+$}
\put(-110,5){${\mathcal{H}_\star}$}
\put(-110,203){${\mathcal{H}^+}$}
\put(-12,202){${\mathcal{C}^+}$}
\put(-12,6){${\mathcal{C}_\star}$}
\end{center}
\caption{Timelike geodesics with $v_\star\neq 0$. All the curves start from the same point on the initial hyperboloid $\mathcal{H}_\star$ and have, from left to right, velocity $v_\star$ given by $\frac{1}{5},\, \frac{1}{3},\, \frac{1}{2},\, \frac{2}{3}, \, \frac{4}{5}, \, \frac{6}{7}$. The diagram is quantitatively correct.}
\label{Figure:TimelikeGeodesics}
\end{figure}

\subsubsection{A blow-up of \texorpdfstring{$\mathcal{C}^+$}{C+}}
A new coordinate which allows to resolve the endpoints of the geodesics terminating at $\mathcal{C}^+$ is given by 
\[
\zeta \equiv \frac{\varsigma \chi}{1-\chi} \quad \Longleftrightarrow \quad \chi = \frac{\zeta}{\zeta +\varsigma }
\]
so that the corner $\mathcal{C}^+$ is now described by
\[
\mathcal{C}^+ =\big\{ \varsigma=0, \; \zeta \in \mathbb{R}, \;\vartheta^\mathcal{A}\in \mathbb{S}^2 \big\}.
\]
Thus, the coordinate $\zeta$ effectively \emph{blows up the corner}. 

\begin{remark}
{\em Indeed, writing the timelike geodesics in terms of the coordinates $(\varsigma, \zeta)$ one has that
\begin{eqnarray*}
&& x(s) = \left(\varsigma(s), \, \zeta(s),\,\vartheta^\mathcal{A}_\star \right) \\
&&  \phantom{ x(s)}= \left(\varsigma(s), \, \frac{\varsigma(s) r(s)}{2\fraka},\,\vartheta^\mathcal{A}_\star \right) \\
&&  \phantom{ x(s)}=\left(\varsigma(s), \frac{\varsigma(s)}{2\fraka}\big( v_\star s +r_\star \big),\,\vartheta^\mathcal{A}_\star \right)\rightarrow \left(0,\frac{v_\star}{2\fraka},\vartheta^\mathcal{A}_\star\right), \qquad \mbox{as} \quad s\rightarrow \infty.
\end{eqnarray*}
Accordingly, the asymptotic value of the coordinate $\zeta$ along a generic timelike geodesic encodes its \emph{velocity}.}
\end{remark}

In terms of the coordinates $(\varsigma,\zeta)$, the Minkowski line element takes the form
\begin{eqnarray*}
&& \tilde\bmeta = \frac{(1+2a \zeta)\Big( (2\zeta +\varsigma)^3 -2 \fraka \zeta \varsigma^2 \Big)}{\varsigma^4 (2\zeta +\varsigma)^2}
\mathbf{d}\varsigma \otimes\mathbf{d}\varsigma \\
&& \hspace{2cm}- \frac{4\fraka \zeta \big(\zeta + \varsigma -\fraka \zeta^2  \big)}{\varsigma^3 (2\zeta + \varsigma)^2} \big(\mathbf{d}\varsigma \otimes \mathbf{d}\zeta + \mathbf{d}\zeta \otimes \mathbf{d}\varsigma  \big)\\
&& \hspace{3cm} -\frac{4\fraka^2}{(2\zeta +\varsigma)^2}\mathbf{d}\zeta \otimes \mathbf{d}\zeta - \frac{4\fraka^2\zeta^2 (\zeta+\varsigma)^2}{\varsigma^2(2\zeta +\varsigma)^2}\slashed\bmsigma.
\end{eqnarray*}

\begin{remark}
    {\em This line element is highly singular at $\varsigma=0$, however, the singular part of the Levi-Civita connection can be removed by means of an order 1 projective extension using the 1-form $\mathbf{d}\varsigma/\varsigma$ ---see equation \ref{ProjectiveExtension:ChristiffleTransformation}. This follows from a lengthy but straightforward calculation.}
\end{remark}

\subsubsection{The conformal rescaled metric}
It is also illuminating to reframe the previous discussion in terms of the conformally rescaled metric \eqref{ConformalMinkowski:B-K}. In terms of the coordinates $(\varsigma,\zeta)$, the line element takes the form
\begin{eqnarray*}
    && \bmeta = \frac{\varsigma^2(1 + 2\fraka \zeta)\big((2\zeta +\varsigma)^2 -2\fraka \zeta \varsigma^2  \big)}{4\fraka^2 (\zeta +\varsigma)^4}\mathbf{d}\varsigma \otimes \mathbf{d}\varsigma  \\
    && \hspace{2cm} - \frac{\zeta \varsigma^3(\zeta +\varsigma - \fraka \varsigma^2)}{\fraka (\zeta + \varsigma)^4}\big(\mathbf{d}\varsigma \otimes \mathbf{d}\zeta + \mathbf{d}\zeta \otimes \mathbf{d}\varsigma  \big) \\
    && \hspace{3cm} - \frac{\varsigma^6}{(\zeta + \varsigma)^4}\mathbf{d}\zeta \otimes \mathbf{d}\zeta - \frac{\varsigma^4 \zeta^2}{(\zeta + \varsigma)^2}\slashed\bmsigma.
\end{eqnarray*}
Observe that this line element is smooth in $(\varsigma,\zeta)$ for $\zeta>0$. Consistently, one can then define a \emph{blow-up} $\mathcal{T}^+$ of $i^+$ as
\[
\mathcal{T}^+ \equiv \big\{ \varsigma =0, \; 0\leq \zeta  \big\}.
\]
This blow-up subsumes the hyperboloid at timelike infinity $\mathcal{H}^+$ and the corner $\mathcal{C}^+$. 

\begin{remark}
{\em The pull-back of the above expression for $\bmeta$ to $\mathcal{T}^+$ vanishes. So, the degeneracy of the metric is more severe at $\mathcal{T}^+$ than at $\mathcal{H}^+$. This is consistent with the form of the conformal factor $\Theta$ as given by equation \eqref{ConformalFactor:Beltrami-KleinCGSystem}. In terms of the coordinates $(\varsigma,\zeta)$, one has that 
\[
\Theta = \frac{\varsigma^3(2\zeta+\varsigma)}{2\fraka (\varsigma + \zeta)^2}.
\]
Consequently, at $\mathcal{T}^+$ the conformal factor vanishes to third order. }
\end{remark}

\begin{remark}
    {\em The coordinate $\zeta$ has a semi-infinite range. This can be compactified by introducing a suitable coordinate transformation $\Acute{\zeta}$ with range, e.g. [0,1]. As an example consider
\[
\Acute{\zeta} = \frac{\fraka \zeta}{1+\fraka \zeta}.
\]    
This new coordinate gives rise to a compactified version of $\mathcal{T}^+$ which we denote by $\Acute{\mathcal{T}}^+$. One then has that $\Acute{\mathcal{T}}^+\simeq \mathbb{B}^3$, where $\mathbb{B}^3$ denotes the 3-dimensional unit ball. The centre of the ball corresponds to $\zeta=0$ ---the endpoint of timelike geodesics with $v_\star=0$, in other words $\mathcal{H}^+$. Further, recall that  $\partial \mathbb{B}^3=\partial\mathbb{S}^2$. This boundary corresponds to the limit $\zeta\rightarrow 0$ which, in turn corresponds to the intersection of $\Acute{\mathcal{T}}^+$ and $\mathscr{I}^+$. These ideas will discussed in detail elsewhere.}
\end{remark}

\subsection{Relation to Ti}
We now connect our construction to the extended boundary $\text{Ti}$ appearing in \cite{FigHavProSal22,BorChaHer25}.  The basic idea is to define a suitable notion of \emph{asymptote} in the sense of Ashtekar.

\begin{figure}[t]
\begin{center}
\includegraphics[width=0.45\textwidth]{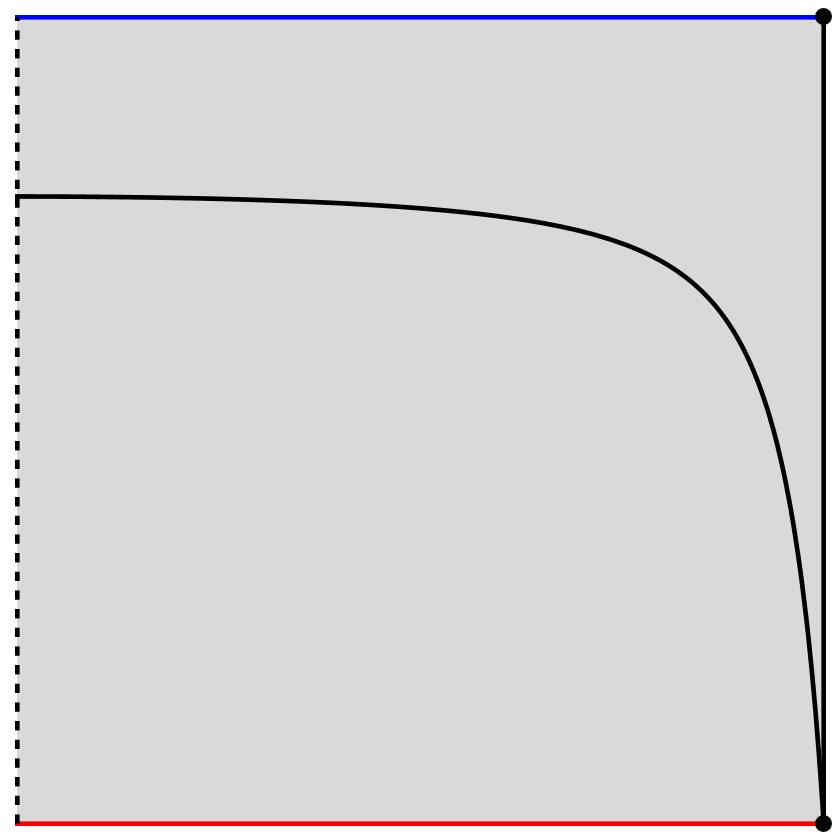}
\put(-1,95){$\scri^+$}
\put(-100,-9){${\mathcal{H}_\star}$}
\put(-100,195){${\mathcal{H}^+}$}
\put(-100,150){${T_{\bullet}}$}
\put(-1,195){${\mathcal{C}^+}$}
\put(-1,-9){${\mathcal{C}_\star}$}
\end{center}
\caption{Diagram showing a constant $T$ hypersurface on the $(\tau,\chi)$ plane. }
\label{Constant-Omega-Surface}
\end{figure}

\subsubsection{Hyperboloidal coordinates in the future light cone}
It is convenient to start from the Minkowski metric in polar coordinates as given by \eqref{Minkowski-spherical} and introduce \emph{hyperbolic coordinates} $(\Psi,X)$ in the interior of the future light cone corresponding to $t>0$ and $t^2>r^2$. The latter are related to the coordinates $(t,r)$ via
\[
  t\equiv \Psi\cosh X,
  \qquad
  r\equiv \Psi\sinh X,
\]
with $\Psi>0, X\in[0,\infty)$. The inverse relations are
\begin{equation}
  \Psi=\sqrt{t^2-r^2},
  \qquad
  \tanh X=\frac{r}{t}.
  \label{Inverse-coordinates-hyperbolic-spherical}
\end{equation}
In terms of these coordinates, the metric becomes
\begin{equation}
  \tilde{\bmeta}= \mathbf{d}\Psi \otimes \mathbf{d}\Psi+\Psi^2 \hat{\bmh}_{\mathbb{H}^{3}},
  \qquad
  \hat{\bmh}_{\mathbb{H}^{3}}\equiv  -\big( \mathbf{d} X \otimes \mathbf{d} X +
  \sinh^2 X\, \slashed \bmsigma\big). 
  \label{eq:hyperbolic-slicing}
\end{equation}
where $\hat{\bmh}_{\mathbb{H}^{3}}$ denotes the unit hyperbolic metric. In these coordinates, future timelike infinity $i^{+}$ is described by the limit $\Psi \to \infty$. To map $i^{+}$ to the origin, introduce $T = \Psi^{-1}$ so that the metric $\tilde{\bmeta}$ becomes
\begin{equation*}
    \tilde{\bmeta} = T^{-4} \mathbf{d} T \otimes \mathbf{d} T + T^{-2} \hat{\bmh}_{\mathbb{H}^{3}}.
\end{equation*}
From \eqref{Inverse-coordinates-hyperbolic-spherical}, we have
\begin{equation*}
    T^{2} = \frac{1}{t^2-r^2},
\end{equation*}
which, in turn, in terms of $(\tau,\chi)$ coordinates can be written as
\begin{equation*}
    T^{2} = \frac{(1-\tau)^2 (1-\chi^2)}{\mu(\tau,\chi)},
\end{equation*}
where for convenience we have introduced the function 
\[
 \mu(\tau,\chi) \equiv(\tau+\fraka(1-\tau))^2-(\tau-\fraka(1-\tau))^2\chi^2
\]
corresponding to the denominator in the previous expression. Observe that $\mu(1,\chi) = (1-\chi^2)$ so that $T^2=O((1-\tau)^2)$. 

\begin{remark}
{\em From the above expression, $T$ vanishes at $\tau=1$ for fixed $\chi<1$. Moreover, for $0<\tau<1$, $T$ also vanishes at $\chi=\pm 1$. In the following analysis, unless otherwise stated, we restrict to $0\leq\chi<1$ and consider the component of $T=0$ obtained in the limit $\tau\to1$ at fixed $\chi<1$. The limiting behaviour as $\chi\to \pm 1$ is not considered here.}
\end{remark}

From the previous expressions, choosing the negative square root of $T^2$, we readily have that 
\begin{equation*}
    \mathbf{d} T = -\frac{(1-\tau) \left( \tau(1 - \chi^2)^2 + \fraka (1 - \tau)
  (1 - \chi^4) \right)}{\mu(\tau,\chi)^2 T} \mathbf{d} \tau - \frac{4 \fraka \tau \chi (1-\tau)^3}{\mu(\tau,\chi)^2 T} \mathbf{d}\chi,
\end{equation*}
so that
\begin{equation*}
    \lim_{\tau \to 1} \mathbf{d} T =  \mathbf{d} \tau.
\end{equation*}
 For the coordinate $X$ on the unit hyperboloid, we have
\begin{equation*}
    \sinh^2X = \frac{4 \fraka^2 (1-\tau)^2\chi^2}{(1-\chi^2)\nu(\tau,\chi)}, \qquad \cosh^2 X=\frac{(\tau (1-\chi^2)+\fraka (1-\tau)(1+\chi^2))^2}{(1-\chi^2)\nu(\tau,\chi)},
\end{equation*}
where 
\begin{equation*}
    \nu(\tau,\chi) \equiv (\tau + \fraka(1 - \tau)(1 - \chi) + \tau\chi) (\tau(1 - \chi) + \fraka(1 - \tau)(1 + \chi)).
\end{equation*}
Since $\nu(1,\chi) = 1-\chi^2$, 
we have $\sinh^2 X = O((1-\tau)^2)$ and $\cosh^2 X = O(1)$. For $\mathbf{d}X$, we get
\begin{equation*}
    \mathbf{d}X = - \frac{2 \fraka \chi}{\nu(\tau,\chi)} \mathbf{d}\tau + \frac{2 \fraka (1-\tau) (\tau (1+\chi^2) + \fraka (1-\tau)(1-\chi^2))}{(1-\chi^2)\nu(\tau,\chi)}\mathbf{d}\chi, 
\end{equation*}
which behaves as
\begin{equation*}
    \lim_{\tau \to 1}\mathbf{d} X = -\frac{2 \fraka \chi}{1-\chi^2} \mathbf{d} \tau.
\end{equation*}

\subsubsection{The rescaled inverse metric}
Now, consider the inverse metric 
\begin{equation*}
    \tilde{\bmeta}^{\sharp} = T^{4} \bmpartial_{T} \otimes \bmpartial_{T} + T^{2} \hat{\bmh}^{\sharp}_{\mathbb{H}^{3}}, \qquad \hat{\bmh}^{\sharp}_{\mathbb{H}^{3}}\equiv - \bmpartial_{X} \otimes \bmpartial_{X} - \csch^2 X \slashed \bmsigma^{\sharp}. 
\end{equation*}
In the following, we will upgrade the coordinate $T$ to a boundary defining function. One can readily verify that 
\begin{equation*}
    \lim_{\tau \to 1}T =0, \qquad \lim_{\tau \to 1}\mathbf{d} T \neq 0.
\end{equation*}
Accordingly, it is natural to define a  \emph{rescaled inverse metric} via 
\begin{equation*}
    \hat{\bmeta}^{\sharp} \equiv  T^{-2} \tilde{\bmeta}^{\sharp} = T^{2} \bmpartial_{T} \otimes \bmpartial_{T} +  \hat{\bmh}^{\sharp}_{\mathbb{H}^{3}}. 
\end{equation*}
The latter extends smoothly to $T =0$.

\subsubsection{Behaviour of the constant \texorpdfstring{$T$}{T} hyperboloids}
 The inverse coordinate transformation between $(T,X)$ and $(\tau,\chi)$ is given by
\begin{equation*}
    \chi = \left( - \fraka T \pm \frac{\sqrt{\cosh 2X -1 + 2 \fraka^2 T^2}}{\sqrt{2}} \right) \csch X. 
\end{equation*}
To keep $\chi > 0$ as $T \to 0$, we pick the positive sign. Solving for $\tau$, we get
\begin{equation}
    \tau = 1+ \frac{\omega(T,X) T}{\cosh 2 X -1 - (1+ \fraka) \omega(T,X)  T - \omega(T,X)  \cosh X},
\end{equation}
where
\begin{equation*}
    \omega(T,X) \equiv  \sqrt{2(\cosh 2X -1 + 2 \fraka^2 T^2)}-2 \fraka T.
\end{equation*}

\begin{remark}
{\em Using these coordinate transformations and setting $T=T_{\bullet}$ where $T_\bullet$ is a constant. One has that $(\tau(T_{\bullet},X),\chi(T_{\bullet},X))$ describes a curve in the $(\tau,\chi)$ plane ---see  Figure \ref{Constant-Omega-Surface}. Note, in particular,  that
\begin{equation*}(\tau(T_{\bullet},X),\chi(T_{\bullet},X))\longrightarrow (0,1) \qquad \mbox{as} \qquad X\rightarrow  \infty.
\end{equation*}
In other words, the hyperboloids of constant $T$ never reach $\mathscr{I}^{+}$. }
\end{remark}

\begin{remark}
{\em Similarly, setting $X=X_{\bullet} \neq 0$, we have 
\begin{equation*}(\tau(T,X_{\bullet}),\chi(T,X_{\bullet}))\longrightarrow (1,1) \qquad \mbox{as} \qquad T\rightarrow  0.
\end{equation*}
Therefore, the constant $X$ hypersurfaces all intersect at the corner $\mathcal{C}^{+}$. }
\end{remark}

\begin{remark}
{\em Using the relation between $(\varsigma,\zeta)$ and $(\tau,\chi)$ discussed earlier, we can show that 
\begin{eqnarray*}
    && T^2 =
    \frac{\varsigma^2(2\zeta+\varsigma)}
    {\bigl(1+2\fraka\zeta+(\fraka-1)\varsigma\bigr)
     \bigl(2\zeta(1-\varsigma)+\varsigma+(\fraka-1)\varsigma^2\bigr)}, \\
    && X = \frac{1}{2}\log \left(
1+\frac{4\fraka\zeta(\zeta+\varsigma)}
{2\zeta(1-\varsigma)+\varsigma+(\fraka-1)\varsigma^2}
\right).
\end{eqnarray*}
These coordinates have the limits 
\begin{eqnarray*}
    && (T,X) \to \left(0, \frac{1}{2} \log (1+ 2 \fraka \zeta) \right) \quad \text{at } \mathcal{C}^{+}, \\
    && (T,X) \to \left(0, 0 \right) \quad \text{at } \mathcal{H}^{+}.
\end{eqnarray*}
In other words, these coordinates encode the same asymptotic information as $(\varsigma,\zeta)$ and as such, can be used to introduce a blow-up of $\mathcal{C}^{+}$ but do not provide valid coordinates for points on $\mathcal{H}^{+}$.
}
\end{remark}

\subsubsection{The rescaled normal}
Now, consider the vectors $\bmpartial_T$ and $\bmpartial_\chi$ expressed in terms of the coordinates $(\tau.\chi)$. A calculation readily gives that 
\begin{equation*}
     \bmpartial_{T} = \frac{\partial \tau}{\partial T} \bmpartial_{\tau} + \frac{\partial \chi}{\partial T} \bmpartial_{\chi}, \qquad \bmpartial_{X} = \frac{\partial \tau}{\partial X} \bmpartial_{\tau} + \frac{\partial \chi}{\partial X} \bmpartial_{\chi},
\end{equation*}
with 
\begin{eqnarray*}
    && \frac{\partial\tau}{\partial T} = \frac{\nu(\tau,\chi) \left(\fraka (1-\tau)(\chi^2-1)+\tau(1+\chi^2)\right)T}{(1-\tau)(1-\chi^2)(1+\chi^2)}, \\
    && \frac{\partial \tau}{\partial X} = \frac{2(1-\tau)\tau\chi}{1+\chi^2}, \qquad \qquad \frac{\partial \chi}{\partial T} = -\frac{\chi \nu(\tau,\chi)  T}{(1-\tau)^2(1+\chi^2)}, \\
    && \frac{\partial \chi}{\partial X}= -\frac{ \tau(1-\chi^2)^2 +\fraka (1-\tau)(1-\chi^4) }{ 2\fraka (1-\tau)(1+\chi^2)}.
\end{eqnarray*}
Now, using the rescaled metric $\hat{\bmeta}$, one can define a \emph{rescaled normal} via
\begin{equation*}
    \bmN \equiv  T^{-2} \hat{\bmeta}^{\sharp}(\mathbf{d} T, \cdot) = \bmpartial_{T}. 
\end{equation*}
Observe, also, that $\langle \mathbf{d}T, \bmN \rangle=1$. The rescaled normal has a finite limit to  $\tau=1$ since
\begin{equation*}
    \lim_{\tau \to 1} \frac{\partial\tau}{\partial T}= -1, \qquad  \left. \frac{\partial \tau}{\partial X} \right|_{\tau=1} =0. 
\end{equation*}
In particular, one has that 
\begin{equation*}
    \lim_{\tau \to 1} \bmN = \lim_{\tau \to 1} \bmpartial_{T} = - \bmpartial_{\tau}.
\end{equation*}
The restriction of $\hat{\bmeta}^{\sharp}$ to the locus of points given by $T =0$ is given by $\hat{\bmh}^{\sharp}_{\mathbb{H}^{3}}$.

\subsubsection{Relation to the conformal Gaussian coordinates}
In order to relate $\hat{\bmh}^{\sharp}_{\mathbb{H}^{3}}$ to the intrinsic metric on the constant $\tau$ hypersurfaces, take the metric $\tilde{\bmeta}$ written in $(\tau,\chi)$ coordinates and define $\grave{\bmeta} = \varpi^2 \tilde{\bmeta}$, so that
\begin{equation*}
    \grave{\bmeta} = \frac{(1-\chi^2)^2}{4 \fraka^2 (1-\tau)^4} \mathbf{d} \tau \otimes \mathbf{d} \tau + \frac{\chi}{\fraka (1-\tau)^2} (\mathbf{d}\tau \otimes
\mathbf{d}\chi + \mathbf{d}\chi \otimes \mathbf{d}\tau ) - \mathbf{d} \chi \otimes \mathbf{d}\chi - \chi^2 \slashed \bmsigma, 
\end{equation*}
with inverse given by
\begin{equation*}
    \grave{\bmeta}^{\sharp} = \frac{4 \fraka^2 (1-\tau)^4}{(1+\chi^2)^2} \bmpartial_{\tau} \otimes \bmpartial_{\tau} + \frac{4 \fraka \chi (1-\tau)^2}{(1+\chi^2)^2} (\bmpartial_{\tau} \otimes \bmpartial_{\chi} + \bmpartial_{\chi} \otimes \bmpartial_{\tau})- \frac{(1-\chi^2)^2}{(1+\chi^2)^2} \bmpartial_{\chi} \otimes \bmpartial_{\chi} - \frac{1}{\chi^2} \slashed \bmsigma^{\sharp}.
\end{equation*}
Note, in particular,  that 
\begin{equation*}
    \grave{\bmh}^{\sharp}_{\mathbb{H}^{3}} \equiv \left. \grave{\bmeta}^{\sharp} \right|_{\tau=1} = - \frac{(1-\chi^2)^2}{(1+\chi^2)^2} \bmpartial_{\chi} \otimes \bmpartial_{\chi} - \frac{1}{\chi^2} \slashed \bmsigma^{\sharp}. 
\end{equation*}
Although $\hat{\bmh}^{\sharp}_{\mathbb{H}^{3}}$ is regular with respect to the coordinate basis $\{ \bmpartial_{T}, \bmpartial_{X} \}$, note that the line element
\begin{eqnarray*}
    && \hat{\bmh}^{\sharp}_{\mathbb{H}^{3}} = -\frac{4(1-\tau)^2\tau^2\chi^2}{(1+\chi^2)^2} \bmpartial_{\tau} \otimes \bmpartial_{\tau} \\
    && \hspace{2cm}+ \frac{\tau\chi( \tau(1-\chi^2)^2 +\fraka (1-\tau)(1-\chi^4)) }{ \fraka (1+\chi^2)^2} (\bmpartial_{\tau} \otimes \bmpartial_{\chi} + \bmpartial_{\chi} \otimes \bmpartial_{\tau}) \\
    && \hspace{3cm}- \frac{(1-\chi^2)^2\left(\tau(1-\chi^2)+\mathfrak a(1-\tau)(1+\chi^2)\right)^2}{4\mathfrak a^2(1-\tau)^2(1+\chi^2)^2} \bmpartial_{\chi} \otimes \bmpartial_{\chi} \\
    &&\hspace{4cm}- \frac{(1-\chi^2)\nu(\tau,\chi)}{4 \fraka^2 (1-\tau)^2\chi^2} \slashed \bmsigma^{\sharp}
\end{eqnarray*}
diverges as $\tau \to 1$.  However, scaling by $\varpi^{-2} \vartheta^2$ (see their definitions in Subsection \ref{Section:Compactification-2-Sphere}) implies a regular metric at $\tau=1$ since 
\begin{equation*}
    \grave{\bmh}^{\sharp}_{\mathbb{H}^{3}} =  \left. (\varpi^{-2} \vartheta^{2} \hat{\bmh}^{\sharp}_{\mathbb{H}^{3}}) \right|_{\tau=1}.
\end{equation*}

\begin{remark}
{\em
Although the metrics $\tilde{\bmh}_{\mathbb H^3}$ and $\hat{\bmh}_{\mathbb H^3}$ have the same standard hyperbolic form, the coordinates $X$ and $\varrho$ should not be identified in the present construction ---$\varrho$ is an intrinsic radial coordinate on the hyperboloids of constant $\psi$ defined in Section \ref{Section:Compactification-2-Sphere} and its relation to the compact radial coordinate $\chi$ is
\[
\chi=\tanh\left(\frac{\varrho}{2}\right),
\qquad
\varrho=2\operatorname{arctanh}\chi.
\]
In other words, it is independent of $\tau$. By contrast,
$X$ is introduced via a spacetime coordinate transformation and is a function of both $\tau$ and $\chi$. In particular, since $\lim_{\tau\rightarrow1}X(\tau,\chi)=0$ for fixed $0\leq\chi<1$, it cannot distinguish different points on $\mathcal{H}^{+}$, whereas $\varrho$, through its dependence on $\chi$, remains a valid radial coordinate there. A similar distinction between intrinsic and spacetime coordinate transformations arises in the comparison of Friedrich’s cylinder with the Ashtekar–Romano hyperboloid at spatial infinity; see Subsections 3.4.1 and 3.4.2 in \cite{MagVal21}.
}
\end{remark}

\begin{remark}
    {\em This discussion gives a direct interpretation of the hyperboloid $\mathcal{H}^{+}$ at $\tau=1$. In the construction introduced in \cite{FigHavProSal22,BorChaHer25}, the unit hyperboloid $T_{\bullet} \simeq \mathbb{H}^{3}$ at $i^+$ provides the base space of the extended Carrollian boundary $\mathrm{Ti}$. Our calculation shows that this base space is related to $\mathcal{H}^{+}$ by a conformal rescaling with the conformal factor given by $\varpi^{-2}\vartheta^2$. The full extended timelike boundary is a line bundle over this base, $\mathrm{Ti}\simeq\mathbb R_u\times\mathbb H^3$,
    where the additional coordinate $u$ distinguishes asymptotic timelike curves with the same limiting direction. }
\end{remark}

\section{Nonlinear perturbations of Minkowski spacetime}
\label{Section:NonlinearPerturbations}

In this section, it is shown that the construction of the hyperboloid at timelike infinity for the Minkowski spacetime discussed in the preceding sections can be extended to suitable nonlinear perturbations. The perturbations of the Minkowski spacetime under consideration are obtained as a solution to an initial value problem with data prescribed on a hyperboloid. The results of this section show the robustness of our construction and open the door to further analysis of the hyperboloidal initial value problem for the vacuum Einstein field equations.

\medskip
We recall the classic semiglobal existence and nonlinear stability of
the Minkowski spacetime ---\cite{Fri86a,LueVal09,CFEBook}:

\begin{theorem}
 \label{Theorem:SemiglobalStabilityMinkowski}
Small enough perturbations of conformally smooth  hyperboloidal initial data for the
Minkowski spacetime give rise to solutions to the vacuum Einstein field equations
which exist globally towards the future, are future null geodesically complete, and have an asymptotic structure similar
to that of the Minkowski spacetime. In particular, the development
admits a smooth conformal extension for which the null generators of
future null infinity $\mathscr{I}^+$ intersect at timelike infinity $i^+$. 
\end{theorem}

\medskip
While the original statement of Friedrich's result only considers the future completeness of null geodesics, the future completeness of timelike geodesics can be readily obtained. Indeed, one has the following corollary:

\begin{corollary}
The semiglobal spacetimes of Theorem
\ref{Theorem:SemiglobalStabilityMinkowski} are future timelike
geodesically complete.
\end{corollary}

\begin{proof}
This follows using the techniques in \cite{ChoCot02}.
\end{proof}

\begin{remark}
{\em The timelike geodesics terminate at $i^+$. }
\end{remark}

The spacetimes $(\tilde{\mathcal{M}},\tilde\bmg)$ obtained from the previous theorem have a a line element of the form
\begin{equation}
\tilde\bmg = \tilde\bmeta + \breve{g}_{\mu\nu}(\bar{x}) \mathbf{d} \bar{x}^\mu \otimes \mathbf{d}\bar{x}^\nu,
\label{PerturbedMetric}
\end{equation}
where $\tilde\bmeta$ is, for example, the Minkwoski line element in the coordinates $\bar{x}=(\bar{x}^\mu)=(\psi,\chi,\vartheta^\mathcal{A})$  ---see equation \eqref{PhysicalMinkowski:psi-chiCoordinates}. The coefficients $\breve{\bmg}_{\mu\nu}(\bar{x})$ decay suitably towards the conformal boundary. Observe that the coordinates $\bar{x}$ of the background spacetime are used to coordinatise the perturbed spacetime. In particular, there exist  a conformal factor $\Xi=\Xi(\bar{x})$ such that the unphysical metric 
\[
\bmg \equiv \Xi^2 \tilde\bmg
\]
is smooth and, moreover, extends smoothly to the locus of points for which $\Xi=0$ ---i.e. future null infinity $\mathscr{I}^+$. 

Now, recall that a congruence of $\tilde\bmeta$-geodesics of the Minkowski spacetime reaching $i^+$ are of the form
\[
x(s)=(s,\chi_\star,\vartheta^\mathcal{A}_\star), \qquad s\in[0,\infty), \quad (\chi_\star,\vartheta^\mathcal{A}_\star)\in \mathcal{H}_\star.
\]
In the perturbed spacetime we prescribe an analogous congruence with initial data
\[
x(0)= (0,\chi_\star,\vartheta^\mathcal{A}_\star), \qquad \dot{\bmx}(0) = \bmpartial_\psi.
\]
The resulting congruence of $\tilde{\bmg}$-geodesics is of the form
\[
x_{(\chi_\star,\vartheta^\mathcal{A}_\star)}(s)= \big(\psi(s,\chi_\star,\vartheta^\mathcal{A}_\star),\chi(s,\chi_\star,\vartheta^\mathcal{A}_\star), \vartheta^\mathcal{A}(s,\chi_\star,\vartheta^\mathcal{A}_\star)\big), \qquad s\in[0,\infty), \quad (\chi_\star,\vartheta^\mathcal{A}_\star) \in\mathcal{H}_\star.
\]

From the stability statement it follows not only that the above curves are future complete but also that the curves in the congruence are non-intersecting. In turn, this implies that the (smooth) map
\[
(s,\chi_\star,\vartheta^\mathcal{A}_\star)\mapsto \big(\psi(s,\chi_\star,\vartheta^\mathcal{A}_\star),\chi(s,\chi_\star,\vartheta^\mathcal{A}_\star), \vartheta^\mathcal{A}(s,\chi_\star,\vartheta^\mathcal{A}_\star)\big)
\]
is invertible with smooth inverse given by
\[
(\psi,\chi,\vartheta^\mathcal{A})\mapsto \big(s(\psi,\chi,\vartheta^\mathcal{A}),\chi_\star(\psi,\chi,\vartheta^\mathcal{A}),\vartheta^\mathcal{A}_\star(\psi,\chi,\vartheta^\mathcal{A})\big).
\]
Thus, 
\[
\grave{x}=(\grave{x}^\mu) =(s,\chi_\star,\vartheta^\mathcal{A}_\star)
\]
defines a system of Gaussian coordinates covering $D^+(\mathcal{H}_\star)$ ---i.e. the future domain of dependence of $\mathcal{H}_\star$. One can thus rewrite the metric \eqref{PerturbedMetric} in terms of the coordinates $\grave{x}$ so as to obtain a line element of the form
\[
\tilde{\bmg} = \grave{g}_{\mu\nu}\mathbf{d}\grave{x}\otimes \mathbf{d}\grave{x},
\]
for some smooth metric coefficients $\grave{g}_{\mu\nu}=\grave{g}_{\mu\nu}(\grave{x})$. 

\medskip
As discussed in Section \ref{Subsection:Reparametrisation} the reparametrisation $s\mapsto\tau=\tau(s)$ that allows to recast timelike geodesics in a vacuum spacetime as conformal geodesics is Universal ---in the sense that it does not depend on the specific form of the metric $\tilde\bmg$. Following the discussion in Section \ref{Section:Compactification-2-Sphere} we thus consider
\[
s=\frac{\tau}{1-\tau}
\]
so that $\tau\rightarrow 1$ as $s\rightarrow \infty$ ---see equation \eqref{Reparameterisation-psi-tau}. Accordingly, writing the metric $\tilde{\bmg}$ in terms of $\tau$ one obtains a line element of the form
\begin{equation}
\tilde{\bmg} = \check{g}_{\mu\nu} \mathbf{d}\check{x}^\mu \otimes \mathbf{d}\check{x}^\nu,
\label{ProjectiveCompactifiedMetric}
\end{equation}
with $\check{g}_{\mu\nu}=\check{g}_{\mu\nu}(\check{x})$ and $\check{x}=(\check{x}^\mu)=(\tau,\chi_\star,\vartheta^\mathcal{A}_\star)$. By construction $\mathcal{H}^+$ is diffeomorphic to $\mathcal{H}_\star$.

\medskip
To complete the discussion we now consider the conformal completion associated to the line element \eqref{ProjectiveCompactifiedMetric}. Again, following the discussion of Section \ref{Subsection:BeltramiKlein} consider a conformal factor of the form
\[
\Theta =(1-\tau)^2 \varpi
\]
where $\varpi$ is the conformal factor that conformally compactifies the asymptotically hyperbolic manifold $\tilde{\mathcal{H}}_\star$ into a manifold with boundary $\mathcal{H}_\star$ diffeomorphic to the interior of $\mathbb{S}^2$ as a subset of $\mathbb{R}^3$. It then follows that the metric 
\begin{eqnarray*}
&& \bmg = \Theta^2\tilde\bmg \\
&& \phantom{\bmg} = (1-\tau)^4\varpi^2  \check{g}_{\mu\nu} \mathbf{d}\check{x}^\mu \otimes \mathbf{d}\check{x}^\nu.
\end{eqnarray*}
This metric is a degenerate but smooth at $\tau=1$. 

\medskip
Using the above results one deduces the following:

\begin{theorem}
\label{Proposition:StabilityGauge}
The  semiglobal spacetimes of Theorem
\ref{Theorem:SemiglobalStabilityMinkowski} admit a conformal
representation with $\mathscr{I}$-fixing for which timelike infinity
is diffeomorphic to the unit solid ball in $\mathbb{R}^3$. 
\end{theorem}

\begin{remark}
    {\em Theorem \ref{Proposition:StabilityGauge} ensures the robustness of the constructions described in the precedent subsection and opens the possibility of using this type of gauge for the analysis of full non-linear perturbations of the Minkowski spacetime by means of $\mathscr{I}$-fixing gauges based on (conformal) Gaussian systems. This will be analysed elsewhere.}
\end{remark}

\begin{remark}
    {\em While the analysis in this section has been based on the Poincar\'e ball gauge, a similar construction can be carried out for the Zenginoglu's gauge.}
\end{remark}

\section{Conclusions and outlook}
\label{Section:Conclusions}
The analysis in the article, which can be regarded of exploratory nature, show that conformal gauges based on the properties of timelike geodesics can be used to provide an extended representation of timelike infinity for spacetimes which are suitable perturbations of the Minkowski spacetime. Crucially, this representation of timelike infinity is different to the one given by the framework of the $\mbox{Ti}^+$ construction. Thus, we show that, conceptually, there is no unique manner of obtaining a blow up of timelike infinity.

Our discussion has focused on spacetimes which are obtained as a solution to an initial value problem on a hyperboloid. The reason for this focus lies on the fact that for these spacetimes, one has available a nonlinear stability result ---accordingly, one has detailed information concerning the asymptotics of the solution. In turn, our construction has been used to provide a succinct analysis of the behaviour of solutions to the wave equations along future null infinity and in a neighbourhood of timelike infinity. Our construction shows close connections to the methods of \emph{b-calculus} ---see e.g. \cite{Mel09,Vas18}. A more ambitious application would be to use the gauges considered in this article to study the conformal Einstein equations for perturbations of Minkowski spacetime. These ideas will be developed elsewhere.

Finally, we also stress that the ideas developed in this article can, in principle, be applied to other vacuum spacetimes like the Schwarzschild and Kerr spacetimes. This, in turn, would provide a new set-up to study the late time behaviour of linear and non-linear perturbations of stationary black holes by means of conformal methods. Again, these ideas will be developed elsewhere.

\acknowledgments
JAVK acknowledges the hospitality of the Erwin Schr\"odinger Institute for Mathematics and Physics of the University of Vienna while attending workshop on \emph{Celestial Holography and Carrollian Physics} in April 2024. The research described in this article was initiated in this period. 

\medskip
MM gratefully acknowledges support from Perimeter Institute for Theoretical Physics and the Fields Institute for Research in Mathematical Sciences through the Fields-AIMS-Perimeter fellowship. Research at Perimeter Institute is supported by the Government of Canada through the Department of Innovation, Science and Economic Development and by the Province of Ontario through the Ministry of Colleges and Universities. This work was supported by the Simons Collaboration on Celestial Holography.  MM acknowledges support by the Simons Foundation grant SFI-MPS-T-Institutes-00010825, and by the State Treasury funds as part of a task commissioned by the Minister of Science and Higher Education under the project Organization of the Simons Semesters at the Banach Center - New Energies in 2026- 2028 (MNiSW/2025/DAP/491). MM thanks the hospitality of the Mathematical Institute of the Polish Academy of Sciences as part of the programme \emph{Simons Programme on Twistor Theory and its Applications} in August/September 2026. 

\medskip

Achilleas Vasileios Koukas gratefully acknowledges the hospitality and support of the Perimeter Institute for Theoretical Physics during his participation in the 2026 PSI Start Research Internship program, where part of this research was conducted.

\medskip
We are grateful to Rod Gover, Emil Have and Yannick Herfray for their helpful comments while working on this manuscript. Calculations in this project used the computer algebra system Mathematica with the package xAct \cite{xAct}. 

\paragraph{Use of AI} The authors declare that the main ideas and major mathematical derivations in this work were developed independently without the use of artificial intelligence tools. GPT-6 Astra was used more locally to assist with standard coordinate changes, in the initial analysis of the transport equations in Appendix \ref{Solution-transport-equations} and with final proofreading. ChatGPT-6.1 Sol was used to assist with literature searches, the revision of the abstract and introduction, and exploratory work on the construction of projective extensions. 

All references and mathematical results were independently verified by the authors and the authors assume full responsibility for the correctness of the final manuscript.


\appendix
\setcounter{equation}{0}  
\renewcommand{\theequation}{\thesection.\arabic{equation}}

\section{Solving the reparametrisation equations}
\label{Appendix:Solving the reparametrisation equations}
For completeness and given that it is of interest on its own right we
provide a derivation of the solution
\eqref{GeneralReparametrisation1}-\eqref{GeneralReparametrisation2}.

\medskip
We begin by observing that equation \eqref{GeneralReparametrisation2},
namely
\[
\dot{\alpha} =\frac{1}{2}\alpha^2 \dot{s}
\]
can be casted in the form
\[
-\left( \frac{1}{\alpha} \right)=\frac{1}{2}\dot{s},
\]
so that after integrating one gets
\[
\frac{1}{\alpha_\star}-\frac{1}{\alpha} = \frac{1}{2}(s-s_\star).
\]
Further rearrangements lead to the expression
\begin{equation}
  \alpha =\frac{2\alpha_\star}{2-\alpha_\star(s-s_\star)}.
  \label{Alpha:FunctionOfs}
\end{equation}  

\medskip
Now, from equation \eqref{ReparametrisationEquation1}, namely 
\[
  \ddot{s} +\alpha \dot{s}^{ 2} =0,
\]
some manipulations readily yield that
\[
\frac{\mathrm{d}}{\mathrm{d}\tau}\left( \frac{1}{\dot{s}} \right) = \alpha.
\]
Thus, using equation  \eqref{Alpha:FunctionOfs} to eliminate $\alpha$
in the last expression, one obtains
\[
\frac{\mathrm{d}}{\mathrm{d}\tau}\left( \frac{1}{\dot{s}}  \right) = \frac{2\alpha_\star}{2-\alpha_\star(s-s_\star)}.
\]
Now, observing that
\begin{eqnarray*}
&& \frac{\mathrm{d}}{\mathrm{d}\tau}\left( \frac{1}{\dot{s}}  \right)
   = \frac{\mathrm{d}}{\mathrm{d}\tau} \left(
   \frac{1}{\mathrm{d}s/\mathrm{d}\tau}\right)\\
  && \phantom{\frac{\mathrm{d}}{\mathrm{d}\tau}\left(
     \frac{1}{\dot{s}}  \right)} =
     \frac{\mathrm{d}}{\mathrm{d}\tau}\left(
     \frac{\mathrm{d}\tau}{\mathrm{d}s} \right)\\
  && \phantom{\frac{\mathrm{d}}{\mathrm{d}\tau}\left(
     \frac{1}{\dot{s}}  \right)} =
     \frac{\mathrm{d}s}{\mathrm{d}\tau}\frac{\mathrm{d}}{\mathrm{d}s}\left(
     \frac{\mathrm{d}\tau}{\mathrm{d}s} \right)\\
  && \phantom{\frac{\mathrm{d}}{\mathrm{d}\tau}\left(
     \frac{1}{\dot{s}}  \right)} =
     \frac{1}{\mathrm{d}\tau/\mathrm{d}s}\frac{\mathrm{d}}{\mathrm{d}s}\left(\frac{\mathrm{d}\tau}{\mathrm{d}s}
     \right) = \frac{\mathrm{d}}{\mathrm{d}s}\ln \left| \frac{\mathrm{d}\tau}{\mathrm{d}s}\right|.
\end{eqnarray*}
Hence, one has that
\[
 \frac{\mathrm{d}}{\mathrm{d}s}\ln \left|
   \frac{\mathrm{d}\tau}{\mathrm{d}s}\right| =\frac{2\alpha_\star}{2-\alpha_\star(s-s_\star)}, 
\]
which, upon integration, yields
\[
\ln \left| \frac{\tau'}{\tau'_\star} \right| = \ln\left( \frac{4}{\big(2-\alpha_\star(s-s_\star)\big)^2} \right).
\]
Taking exponentials on both sides and observing that the resulting
right-hand side is positive one can remove the absolute value on the
left-hand side so as to obtain
\[
\tau' = \frac{4 \tau'_\star}{\big( (2+\alpha_\star s_\star)-\alpha_\star s \big)^2}.
\]
A final integration then gives
\[
\tau -\tau_\star = \frac{2\tau'_\star(s-s_\star)}{2+\alpha_\star(
  s_\star- s)}.
\]
Inverting this last equations one then finds
\[
s-s_\star = \frac{2(\tau-\tau_\star)}{2\tau'_\star +\alpha_\star(\tau-\tau_\star)}.
\]
Finally, substituting this expression into equation
\eqref{Alpha:FunctionOfs} one obtains an expression for $\alpha$ as a
function of $\tau$ ---namely
\[
\alpha-\alpha_\star =
\frac{\alpha^2_\star}{2\tau'_\star}(\tau-\tau_\star). 
\]

\section{The NP spin coefficients}
\label{Appendix:NPComputations}
In this appendix we provide further details of the computation of the NP spin coefficients of the varoius conformal representations of the Minkowski spacetimes considered in the main text.

\subsection{A NP-frame for the Poincar\'e ball gauge}
\label{SubSection:NPFrame}

Starting with the contravariant metric $\bmeta^\sharp$ as given by equation \eqref{EtaSharp} and completing the squares, we can bring it into the form
\begin{multline*}
\bmeta^\sharp = \frac{4\fraka^2}{(1-\chi^2)^2}\bmpartial_\tau\otimes\bmpartial_\tau \\-\frac{(1-\chi^2)^2}{(1+\chi^2)^2}\left(\frac{1}{(1-\tau)^2}\bmpartial_\chi - \frac{4\fraka\chi}{(1-\chi^2)^2}\bmpartial_\tau \right) \otimes \left( \frac{1}{(1-\tau)^2}\bmpartial_\chi - \frac{4\fraka\chi}{(1-\chi^2)^2}\bmpartial_\tau\right) \\
- \frac{1}{\chi^2(1-\tau)^4}\slashed\bmsigma^\sharp.\end{multline*}

Thus, one can read the orthonormal vectors 
\begin{eqnarray*}
&& \bme_\bmzero = \frac{2\fraka}{1-\chi^2}\bmpartial_\tau , \qquad \bme_\bmone = \frac{1-\chi^2}{(1-\tau)^2(1+\chi^2)}\bmpartial_\chi-\frac{4\fraka \chi }{1-\chi^4}\bmpartial_\tau.
\end{eqnarray*}

One can verify that 
\[
\bmeta(\bme_\bmzero,\bme_\bmzero) =1, \qquad \bmeta(\bme_\bmone,\bme_\bmone) =-1, \qquad \bmeta(\bme_\bmzero,\bme_\bmone)=0.
\]

\medskip
Now, a NP-frame $\{\bml,\,\bmn, \, \bmm,\, \bar\bmm \}$ satisfies
\[
\bmeta^\sharp = \bml\otimes\bmn +\bmn\otimes\bml - \bmm\otimes \bar\bmm -\bar\bmm \otimes \bmm,
\]
where the only non-vanishing products are given by
\[
\bmeta(\bml,\bmn)=1, \qquad \bmeta(\bmm,\bar\bmm)=-1.
\]
Given $\bme_\bmzero$, $\bme_\bmone$ as above one has that
\[
\bml = \frac{1}{\sqrt{2}}(\bme_\bmzero+\bme_\bmone), \qquad \bmn=\frac{1}{\sqrt{2}}(\bme_\bmzero-\bme_\bmone).
\]
Thus, a quick computation gives that 
\begin{subequations}
    \begin{eqnarray}
        && \bml =\frac{\sqrt{2}\fraka(1-\chi)^2}{1-\chi^4}\bmpartial_\tau +\frac{1-\chi^2}{\sqrt{2}(1-\tau)^2(1+\chi^2)}\bmpartial_\chi,   \label{NP-tetrad-first-choice_a} \\
        && \bmn =\frac{\sqrt{2}\fraka(1+\chi)^2}{1-\chi^4}\bmpartial_\tau -\frac{1-\chi^2}{\sqrt{2}(1-\tau)^2(1+\chi^2)}\bmpartial_\chi. 
          \label{NP-tetrad-first-choice_b}
    \end{eqnarray}
  
\end{subequations}
Also, we have that 
\[
\bmm\otimes \bar\bmm +\bar\bmm \otimes \bmm=\frac{1}{\chi^2(1-\tau)^4}\slashed\bmsigma^\sharp.
\]
Thus, one can set
\begin{subequations}
\begin{eqnarray}
\bmm = \frac{1}{\chi(1-\tau)^2}\bmpartial_+, \label{NP-tetrad-first-choice_c}\\
\bar\bmm = \frac{1}{\chi(1-\tau)^2}\bmpartial_-,
\label{NP-tetrad-first-choice_d}
\end{eqnarray}
\end{subequations}
where $\{\bmpartial_+,\, \bmpartial_-\}$ span $T\mathbb{S}^2$. In particular, we require
\[
\slashed\bmsigma^\sharp = \bmpartial_+\otimes\bmpartial_-+\bmpartial_-\otimes\bmpartial_+.
\]

Following the notation of \cite{Ste91}, the NP connection coefficients $(\upkappa,\uptau,\upsigma,\uprho,\upepsilon,\upgamma,\upbeta,\upalpha,\uppi,\upnu,\upmu,\uplambda)$ can be calculated using the choice of the tetrad 
\eqref{NP-tetrad-first-choice_a}-\eqref{NP-tetrad-first-choice_b}, \eqref{NP-tetrad-first-choice_c}-\eqref{NP-tetrad-first-choice_d}
 and the commutator relations in page 216 of \cite{Ste91}, we get
\begin{eqnarray*}
    && \upkappa =0, \qquad \uptau =0, \qquad \upsigma =0, \qquad \uprho = -\frac{(1-\chi)(1+\chi(2+\chi-4 \fraka(1-\tau)))}{\sqrt{2} \chi (1+\chi) (1+\chi^2) (1-\tau)^2}, \\
    && \upepsilon = \frac{ 2 \fraka (\chi -1 + \tau - \chi \tau)-\chi (1+\chi)}{\sqrt{2} (1+\chi) (1+\chi^2) (1-\tau)^2}, \qquad \upgamma = -\frac{\chi - \chi^2- 2 \fraka (1+\chi) (1-\tau)}{\sqrt{2} (1-\chi) (1+\chi^2) (1-\tau)^2}, \\
    && \upbeta = \frac{\cot \theta}{2 \sqrt{2} \chi (1-\tau)^2}, \qquad \upalpha = - \frac{\cot \theta}{2 \sqrt{2} \chi (1-\tau)^2}, \qquad \uppi =0, \qquad \upnu =0, \\
    && \upmu = -\frac{(1+\chi)(1+\chi(\chi-2+4 \fraka(1-\tau)))}{\sqrt{2} \chi (1-\chi) (1+\chi^2) (1-\tau)^2}, \qquad \uplambda =0.
\end{eqnarray*}

In the main text, we adopt the compact NP notation, in which the following relabelling of the spin coefficients is made:
\[
\uppi \mapsto -\uptau', \qquad
\upmu \mapsto -\uprho', \qquad
\uplambda \mapsto -\upsigma', \qquad
\upnu \mapsto -\upkappa'.
\]

\subsection{The NP-frame for the Zenginoglu gauge}
Similarly, for the contravariant metric $\bar{\bmeta}^\sharp$ as given by \eqref{EtaBarSharp} we can write 
\begin{multline*}
\bar{\bmeta}^{\sharp} = \frac{\fraka^2}{\fraka^2 (1-\rho)^2 +\rho^2} \bmpartial_{\tau} \otimes \bmpartial_{\tau} - \frac{\rho}{(1-\rho) (1-\tau)^2 \sqrt{\fraka^2 + \tfrac{\rho^2}{(1-\rho)^2}}} \left( \bmpartial_{\tau} \otimes \bmpartial_{\rho} + \bmpartial_{\rho} \otimes \bmpartial_{\tau}  \right) \\
- \frac{(1-\rho)^2}{(1-\tau)^4} \bmpartial_{\rho} \otimes \bmpartial_{\rho} - \frac{1}{\rho^2 (1-\tau)^4} \slashed\bmsigma^{\sharp}.
\end{multline*}
As such, we can introduce the orthonormal vectors
\begin{eqnarray*}
    && \bar{\bme}_{\bmzero} = \frac{1}{\Omega} \bmpartial_{\tau}, \; \; \qquad \bar{\bme}_{\bmone}  = \frac{\Omega}{(1-\tau)^2(\Omega-\rho \Omega')} (\bmpartial_{\rho} - (1-\tau)^2 F'(\rho) \bmpartial_{\tau}) \\
    && \phantom{\bar{\bme}_{\bmzero}} = \frac{1}{1-\rho} \bmpartial_{\tau}, \quad \phantom{\bar{\bme}_{\bmone}} = \frac{1-\rho}{(1-\tau)^2} \bmpartial_{\rho} - \frac{\fraka^2}{(1-\rho) \left(\fraka^2 + \frac{\rho^2}{(1-\rho)^2} \right)^{\tfrac{3}{2}}} \bmpartial_{\tau},
\end{eqnarray*}
and 
\begin{equation*}
    \bar{\bme}_{\bmtwo} = \frac{1}{\rho (1-\tau)^2} \bmpartial_\theta, \qquad \bar{\bme}_{\bmthree} = \frac{\sin \theta}{\rho (1-\tau)^2} \bmpartial_{\varphi}.
\end{equation*}

In a slight abuse of notation, let $\{ \bml, \bmn, \bmm, \bar{\bmm} \}$ denote the NP null tetrad associated with $\bar{\bmeta}$ which is defined by
\begin{equation*}
    \bml = \frac{1}{\sqrt{2}} (\bar{\bme}_{\bmzero} + \bar{\bme}_{\bmone}), \qquad \bmn= \frac{1}{\sqrt{2}} (\bar{\bme}_{\bmzero} - \bar{\bme}_{\bmone}), \qquad \bmm = \frac{1}{\sqrt{2}} (\bar{\bme}_{\bmtwo} + i \bar{\bme}_{\bmthree}), \qquad \bmm = \frac{1}{\sqrt{2}} (\bar{\bme}_{\bmtwo} - i \bar{\bme}_{\bmthree})
\end{equation*}
which can be written explicitly as 
\begin{subequations}
    \begin{eqnarray}
        && \bml =  \frac{\Omega-\Omega^2 f'(\rho)-\rho \Omega'}{\sqrt{2} \Omega (\Omega-\rho \Omega')} \bmpartial_{\tau}+ \frac{\Omega}{\sqrt{2} (1-\tau)^2 (\Omega-\rho \Omega')} \bmpartial_{\rho},  \label{NP-tetrad-second-choice_a}\\
        && \bmn =\frac{\Omega+\Omega^2 f'(\rho)-\rho \Omega'}{\sqrt{2} \Omega (\Omega-\rho \Omega')} \bmpartial_{\tau} -\frac{\Omega}{\sqrt{2} (1-\tau)^2 (\Omega-\rho \Omega')} \bmpartial_{\rho}.  \label{NP-tetrad-second-choice_b}
    \end{eqnarray}
\end{subequations}
Furthermore, we can choose complex frames $\{\bmpartial_+,\, \bmpartial_-\}$ that span $T\mathbb{S}^2$ so that
\begin{equation}
    \bmm = \frac{1}{\rho (1-\tau)^2} \bmpartial_{+}, \qquad \bar{\bmm} = \frac{1}{\rho (1-\tau)^2} \bmpartial_{-}.
    \label{NP-tetrad-second-choice_c}
\end{equation}

For the frame given by \eqref{NP-tetrad-second-choice_a}-\eqref{NP-tetrad-second-choice_b} and \eqref{NP-tetrad-second-choice_c} it follows then that the associated connection coefficients are given by 
\begin{eqnarray*}
    && \upkappa =0, \qquad \uptau =0, \qquad \upsigma =0, \qquad \uprho = -\frac{R \left( 1+ \rho^2 - 2 \rho (\tau-2) \right) + 2 \rho^2 (1-\tau)}{\sqrt{2} R \rho (1-\rho) (1-\tau)^2}, \\
    && \upepsilon = \frac{2 \rho (1-\tau) + R (\rho-3 + 2\tau)}{2 \sqrt{2} R (1-\rho) (1-\tau)^2}, \qquad \upgamma = \frac{R (1+\rho - 2 \tau) + 2 \rho (1-\tau)}{ 2 \sqrt{2} R(1-\rho) (1-\tau)^2}, \\
    && \upbeta = \frac{\cot \theta}{2 \sqrt{2} \rho (1-\tau)^2}, \qquad \upalpha = - \frac{\cot \theta}{2 \sqrt{2} \rho (1-\tau)^2}, \qquad \uppi =0, \qquad \upnu =0, \\
    && \upmu = -\frac{2 \rho^2 (1-\tau)+ R (1+\rho^2 - \rho \tau)}{\sqrt{2} R  \rho (1-\rho) (1-\tau)^2}, \qquad \uplambda =0,
\end{eqnarray*}
where $R$ is the function of $\rho$ defined in the main text ---see equation \eqref{Definition-R}.

\section{Solving the transport equations near \texorpdfstring{$\tau=1$}{tau=1}}
\label{Solution-transport-equations}
In this Appendix, we provide further details on the asymptotic behaviour of the solutions to the transport equations on $\mathscr{I}^+$ implied by the wave equation.

\subsection{Poincar\'e ball gauge}
In this section, we provide an overview of the arguments leading to Lemma \eqref{Lemma:BehaviourCPlus}.

\medskip
In the main text we consider the ODE
\begin{equation}
    (1-\tau)^4 \ddot{\phi}_{\ell} + \frac{(1-\tau)^2}{\fraka^2} (2-6\fraka^2 (1-\tau)) \dot{\phi}_{\ell} + \frac{\ell (\ell +1)}{\fraka^2} \phi_{\ell} =0
    \label{ODE-homogeonous-phi-l}
\end{equation}
---cfr. \eqref{ODE:Poincare}. A more detailed version of Lemma \ref{Lemma:BehaviourCPlus} is given by the following:

\begin{proposition}
Let $\phi_{\ell}$ be a solution of the homogeneous ODE
\eqref{ODE-homogeonous-phi-l}, and define
\[
D_{\ell}^{2}\equiv 1-\fraka^{2}\ell(\ell+1).
\]
Assume that $\phi_{\ell}(\tau)$ can be expanded near $\tau=1$ as
\begin{equation*}
    \phi_{\ell} (\tau) = e^{\tfrac{\mu^{(0)}_{\ell}}{1-\tau}} ((1-\tau)^{\mu^{(1)}_{\ell}}+ O((1-\tau)^{\mu^{(1)}_{\ell}+1}),
\end{equation*}
and that $D_{\ell}^{2}\neq0$. Then, as $\tau\to1^{-}$:
\begin{enumerate}
    \item[(i)] For $\ell=0$, the solution $\phi_{0}$ extends smoothly to $\tau=1$, with $\phi_{0}(1)=A_{0}$.

    \item[(ii)] If $\ell>0$ and $D_{\ell}^{2}>0$, then both independent solutions decay exponentially and are flat at $\tau=1$. Consequently, every solution extends smoothly to $\tau=1$, with $\phi_{\ell}(1)=0$.

    \item[(iii)] If $D_{\ell}^{2}<0$, then the solution is rapidly oscillatory with an exponentially decaying amplitude and extends smoothly to $\tau=1$, with $\phi_{\ell}(1)=0$.
\end{enumerate}
In cases \textnormal{(ii)} and \textnormal{(iii)}, the smoothness
statements assume that the lower-order terms admit smooth asymptotic
expansions.
\end{proposition}

\subsubsection*{Details}
\label{Solution-transport-equation}

Starting from \eqref{ODE-homogeonous-phi-l}, introducing $x\equiv 1-\tau$ and dividing by $x^4$, we get
\begin{equation}
    \ddot{\phi}_{\ell} + \left( \frac{6}{x} - \frac{2}{\fraka^2 x^2}\right) \dot{\phi}_{\ell} + \frac{\ell (\ell +1)}{\fraka^2 x^4}\phi_{\ell} =0, 
    \label{ODE-alt-compact-wave-equation}
\end{equation}
where, in an abuse of notation, we used $\dot{(\; \;)}$ to denote differentiation with respect to $x$. For a second order homogenous ODE of the form
\begin{equation*}
    \ddot{\phi} + P(x) \dot{\phi} + Q(x) \phi =0,
\end{equation*}
a point $x=x_{0}$ is singular if either $P(x)$ or $Q(x)$ is not analytic at $x_{0}$. This point is a regular singularity if $(x-x_{0}) P(x)$ and $(x-x_{0})^2 Q(x)$ are analytic. In our case, it is clear that $x=0$ is a \emph{an irregular} singular point. In this case, a suitable Ansatz is given by
\begin{equation*}
    \phi_{\ell} (x) = e^{\mu_{\ell}(x)},
\end{equation*}
with $\mu_{\ell}(x)$ expanded near $x=0$ as
\begin{equation*}
    \mu_{\ell}(x) = \frac{\mu^{(0)}_{\ell}}{x} + \mu^{(1)}_{\ell} \log{x} + \mathrm{l.o.t},
\end{equation*}
so that $\phi_{\ell} (x)$ can be written near $x=0$ or $(\tau=1)$ as
\begin{equation}
    \phi_{\ell} (\tau) = e^{\tfrac{\mu^{(0)}_{\ell}}{1-\tau}} ((1-\tau)^{\mu^{(1)}_{\ell}}+ O((1-\tau)^{\mu^{(1)}_{\ell}+1}).
    \label{Solution-phi-l-other-compactification}
\end{equation}
Substituting the leading order contribution in \eqref{ODE-alt-compact-wave-equation} and solving for $\mu^{(0)}_{\ell}$ gives
\begin{equation*}
    \mu^{(0)}_{\ell \pm} = \frac{-1\pm \sqrt{1- \fraka^2 \ell(\ell+1)}}{\fraka^2}.
\end{equation*}
The corresponding solutions for $\mu^{(1)}_{\ell\pm}$ in these cases are given by
    \begin{equation}
        \mu^{(1)}_{\ell \pm} = - \frac{2( \sqrt{1- \fraka^2 \ell(\ell+1)} \mp 1 )}{\sqrt{1- \fraka^2 \ell(\ell+1)}}.
        \label{mu-1-l}
    \end{equation}

\medskip
At this point we distinguish two cases:

\begin{itemize}
    \item \textbf{For $\ell =0$}: 
    \begin{equation*}
    \mu^{(0)}_{0 +} = 0,  \qquad \text{or } \quad \mu^{(0)}_{0 -} = -\frac{2}{\fraka^2},
    \end{equation*}
    and $\mu^{(1)}_{\ell \pm}$ evaluates to
    \begin{equation*}
        \mu^{(1)}_{0+} = 0,  \qquad \text{or } \quad \mu^{(1)}_{0-} = -4.
    \end{equation*}
    From the above, we can write $\phi_{0}$ as
    \begin{equation*}
        \phi_{0} (\tau) = A_{0} + B_{0}  e^{-\tfrac{2}{\fraka^2 (1-\tau)}} \left( \frac{1}{(1-\tau)^4} + O\left((1-\tau)^{-3} \right)\right),
    \end{equation*}
    where $A_{0}$ and $B_{0}$ are some constants that can be determined from initial data. Since the exponential factor decays faster than any inverse power of $1-\tau$, the second term vanishes as $\tau\to1^{-}$. Moreover, this term is flat at $\tau=1$, in the sense that all of its derivatives vanish there. Consequently, the solution admits a smooth extension to $\tau=1$, with
\begin{equation}
    \phi_{0}(1)=A_{0}.
\end{equation}

\item \textbf{For general $\ell>0$}: we can write the general solution as 
\begin{equation}
    \phi_{\ell}(\tau) = A_{\ell}  \phi_{\ell +}(\tau)+ B_{\ell}  \phi_{\ell -}(\tau),
    \label{Solution-l-bigger-than0}
\end{equation}
where 
\begin{eqnarray*}
    && \phi_{\ell +}(\tau)=e^{\tfrac{\mu^{(0)}_{\ell+}}{1-\tau}} ((1-\tau)^{\mu^{(1)}_{\ell+}}+ O((1-\tau)^{\mu^{(1)}_{\ell+}+1}), \\
    && \phi_{\ell -}(\tau)=e^{\tfrac{\mu^{(0)}_{\ell-}}{1-\tau}} ((1-\tau)^{\mu^{(1)}_{\ell-}}+ O((1-\tau)^{\mu^{(1)}_{\ell-}+1}).
\end{eqnarray*}
Note that coefficients $\mu^{(0)}_{\ell\pm}$ are real if $D^2_{\ell} >0$ and form a complex-conjugate pair if $D_\ell^2<0$.

\begin{table}[H]
\centering
\renewcommand{\arraystretch}{1.15}
\setlength{\tabcolsep}{7pt}
\begin{tabular}{ccccc}
\toprule
\textbf{Condition}
& $\boldsymbol{\mu^{(0)}_{\ell+}}$
& $\boldsymbol{\mu^{(1)}_{\ell+}}$
& $\boldsymbol{\mu^{(0)}_{\ell-}}$
& $\boldsymbol{\mu^{(1)}_{\ell-}}$ \\
\midrule
$0<D_\ell<1$
& $<0$
& $>0$
& $<0$
& $<0$ \\
\bottomrule
\end{tabular}
\caption{Signs of the exponents when $D_\ell^2>0$ and $\ell>0$.}
\label{Table:signs-exponents-positive-D}
\end{table}

For real $\mathfrak a\neq0$ and $\ell>0$, the condition
$D_\ell^2>0$ implies $0<D_\ell<1$. In this case,
$\mu^{(0)}_{\ell+}<0$ and $\mu^{(0)}_{\ell-}<0$, so the exponential
factors in both independent solutions decay faster than any power of
$1-\tau$. Hence, provided that the lower-order terms admit smooth
asymptotic expansions, both $\phi_{\ell+}$ and $\phi_{\ell-}$ extend
smoothly to $\tau=1$ by setting
\[
\phi_{\ell+}(1)=\phi_{\ell-}(1)=0.
\]
Consequently, every
solution of the homogeneous equation with $D_\ell^2>0$ and $\ell>0$
satisfies
\[
\lim_{\tau\to1^-}\phi_\ell(\tau)=0.
\]

\begin{remark}
{\em The case $D_{\ell}=1$ occurs if and only if $\ell=0$ (or if $\fraka=0$ which we assume to be non-zero). Hence, it does not give rise to an additional class of asymptotic behaviour beyond the monopole case considered above.}
\end{remark}

If $D^2_{\ell}<0$, then we have $D_{\ell} \in i \mathbb{R}$ and we can write
\begin{equation*}
    \mu^{(0)}_{\ell \pm} = \frac{-1\pm i \sigma \delta_{\ell}}{\fraka^2} =  \frac{-1\pm i  \delta_{\ell}}{\fraka^2},
\end{equation*}
where $\delta_{\ell} \equiv \sqrt{|1-\fraka^2 \ell(\ell+1)|}$ and $\sigma = \pm1$. We can also write
\begin{equation}
        \mu^{(1)}_{\ell \pm} = - \frac{2 ( \delta_{\ell} \pm i \sigma)}{ \delta_{\ell}} = - \frac{2 ( \delta_{\ell} \pm i)}{ \delta_{\ell}}, 
\end{equation}
Using these expressions in \eqref{Solution-l-bigger-than0}, we can write $\phi_{\ell}(\tau)$ as 
\begin{equation}
    \phi_{\ell}(\tau) = f(\tau) \left( A_{\ell} e^{i \Phi_{\ell}(\tau)} + B_{\ell} e^{-i \Phi_{\ell}(\tau)} \right)+ \mathrm{l.o.t},
    \label{Solution-l-bigger-than0-complex}
\end{equation}
where
\begin{eqnarray*}
    && f(\tau) = \frac{e^{-\frac{1}{\fraka^2 (1-\tau)}}}{(1-\tau)^2}, \qquad \Phi_{\ell}(\tau)= \frac{\delta_{\ell}}{\fraka^2 (1-\tau)} - \frac{2}{\delta_{\ell}} \log(1-\tau). 
\end{eqnarray*}
As $\tau\to1^{-}$, the amplitude $f(\tau)$ decays exponentially to
zero, since the exponential decay dominates the factor
$(1-\tau)^{-2}$, while the phase $\Phi_{\ell}(\tau)$ becomes
unbounded. Consequently, the leading-order solution consists of two
increasingly rapidly oscillating modes with a common exponentially
decaying amplitude.
\end{itemize}

\begin{remark}
{\em The critical case $D_{\ell}=0$ corresponds to a repeated leading-order root and requires a separate asymptotic analysis.}
\end{remark}




\bibliography{Newgrbib}
\bibliographystyle{JHEP}

\end{document}